\documentclass[11pt]{article}
\usepackage{pdfpages}
\usepackage{multirow}
\usepackage{array}
\usepackage{graphicx}
\usepackage{subcaption}
\usepackage{latexsym}
\usepackage{pstricks}
\usepackage{pst-node}
\usepackage{pst-tree}
\usepackage{url}
\usepackage{times}
\usepackage{helvet}
\usepackage{courier}
\usepackage{algorithmic}
\usepackage[ruled,vlined,linesnumbered]{algorithm2e}
\usepackage{amsthm}
\usepackage{graphicx,amssymb,amsmath}
\usepackage{mathrsfs}
\usepackage{mathtools}
\usepackage{comment}
\usepackage{framed}

\DeclareMathOperator*{\argmax}{arg\,max}

\DeclareMathOperator*{\E}{\mathbb{E}}

\usepackage[hypertexnames=false, colorlinks=true, citecolor=blue, linkcolor=red, urlcolor=black]{hyperref}

\usepackage{geometry}
\makeatletter
\renewcommand*{\verbatim@font}{\sffamily}

\title{Connected Components on a PRAM in Log Diameter Time\thanks{A preliminary version of this paper appeared in the proceedings of SPAA 2020.}}

\author{S. Cliff Liu\thanks{Research at Princeton University partially supported by an innovation research grant from Princeton University and a gift from Microsoft. Some work was done at AlgoPARC workshop in 2019, partially supported by NSF grant CCF-1745331.} \\
Princeton University	\\
\textsf{sixuel@cs.princeton.edu}
\and
Robert E. Tarjan\footnotemark[2] \\
Princeton University \\
\textsf{ret@cs.princeton.edu}
\and
Peilin Zhong\thanks{Supported in part by NSF grants (CCF-1703925, CCF-1714818, CCF-1617955 and CCF-1740833, CCF-1822809), Simons Foundation (\#491119 to Alexandr Andoni), Google Research Award and a Google Ph.D. fellowship.} \\
Columbia University \\
\textsf{peilin.zhong@columbia.edu}
}

\begin{document}

\clearpage\maketitle

\thispagestyle{empty}

\newcounter{dummy} \numberwithin{dummy}{section}
\newtheorem{lemma}[dummy]{Lemma}
\newtheorem{definition}[dummy]{Definition}
\newtheorem{remark}[dummy]{Remark}
\newtheorem{corollary}[dummy]{Corollary}
\newtheorem{claim}[dummy]{Claim}
\newtheorem{observation}[dummy]{Observation}
\newtheorem{assumption}[dummy]{Assumption}

\newcounter{dummy2}
\newtheorem{theorem}[dummy2]{Theorem}

\begin{abstract}
    We present an $O(\log d + \log\log_{m/n} n)$-time randomized PRAM algorithm for computing the connected components of an $n$-vertex, $m$-edge undirected graph with maximum component diameter $d$. The algorithm runs on an ARBITRARY CRCW (concurrent-read, concurrent-write with arbitrary write resolution) PRAM using $O(m)$ processors. The time bound holds with good probability.\footnote{To simplify the statements in this paper we assume $m/2 \ge n \ge 2$ and $d \ge 1$. \emph{With good probability} means with probability at least $1 - 1 / \text{poly}((m \log n) / n)$; \emph{with high probability} means with probability at least $1 - 1 / \text{poly}(n)$.}

    Our algorithm is based on the breakthrough results of Andoni et al. [FOCS'18] and Behnezhad et al. [FOCS'19]. Their algorithms run on the more powerful MPC model and rely on sorting and computing prefix sums in $O(1)$ time, tasks that take $\Omega(\log n / \log\log n)$ time on a CRCW PRAM with $\text{poly}(n)$ processors. Our simpler algorithm uses limited-collision hashing and does not sort or do prefix sums. It matches the time and space bounds of the algorithm of  Behnezhad et al., who improved the time bound of Andoni et al.

    It is widely believed that the larger private memory per processor and unbounded local computation of the MPC model admit algorithms faster than that on a PRAM. Our result suggests that such additional power might not be necessary, at least for fundamental graph problems like connected components and spanning forest.
\end{abstract}

\newpage

\pagenumbering{roman}

{\hypersetup{linkcolor=black, linktoc=all}\tableofcontents}

\newpage

\pagenumbering{arabic}
\setcounter{page}{1}

\section{Introduction}\label{intro}

Computing the connected components of an undirected graph is a fundamental problem in algorithmic graph theory, with many applications.
Using graph search \cite{DBLP:journals/siamcomp/Tarjan72}, one can find the connected components of an $n$-vertex, $m$-edge graph in $O(m)$ time, which is best possible for a sequential algorithm. But linear time is not fast enough for big-data applications in which the problem graph is of internet scale or even bigger.  To find the components of such large graphs in practice requires the use of concurrency.

Beginning in the 1970's, theoreticians developed a series of more-and-more efficient concurrent algorithms. Their model of computation was some variant of the PRAM (parallel random access machine) model.
Shiloach and Vishkin \cite{DBLP:journals/jal/ShiloachV82}  gave an $O(\log n)$-time PRAM algorithm in 1982.
The algorithm was later simplified by Awerbuch and Shiloach \cite{DBLP:journals/tc/AwerbuchS87}.
These algorithms are deterministic and run on an ARBITRARY CRCW PRAM with $O(m)$ processors.
%Their algorithm and analysis are .
%Both algorithms above are deterministic.
There are simpler algorithms and algorithms that do less work but use randomization.
%Reif proposed an elegant randomized algorithm under the leader contraction framework \cite{reif1984optimal}, with the same
%It uses $O(m)$ processors and has $O(\log n)$ running time on a CRCW PRAM. %, which will be discussed in details later.
Gazit \cite{DBLP:journals/siamcomp/Gazit91} combined an elegant randomized algorithm of Reif \cite{reif1984optimal} with graph size reduction to obtain an $O(\log n)$-time, $O(m / \log n)$-processor CRCW PRAM algorithm.
This line of work culminated in the $O(\log n)$-time, $O(m / \log n)$-processor EREW (exclusive-read, exclusive-write) PRAM algorithms of Halperin and Zwick \cite{DBLP:journals/jcss/HalperinZ96, halperin2001optimal}, the second of which computes spanning trees of the components as well as the components themselves. On an EREW PRAM, finding connected components takes $\Omega(\log n)$ time \cite{DBLP:journals/siamcomp/CookDR86}, so these algorithms minimize both time and work.
The $\Omega(\log n)$ lower bound also holds for the CREW (concurrent-read, exclusive-write) PRAM with randomization, a model slightly weaker than the CRCW PRAM \cite{DBLP:journals/jcss/DietzfelbingerKR94}.
For the PRIORITY CRCW PRAM (write resolution by processor priority), a time bound of $\Omega(\log n / \log\log n)$ holds if there are $\text{poly}(n)$ processors and unbounded space, or if there is $\text{poly}(n)$ space and any number of processors \cite{DBLP:journals/jacm/BeameH89}.

The Halperin-Zwick algorithms use sophisticated techniques.
Practitioners charged with actually finding the connected components of huge graphs have implemented much simpler algorithms.
Indeed, such simple algorithms often perform well in practice \cite{DBLP:conf/dimacs/GoddardKP94, DBLP:conf/spaa/Greiner94, hsu1997parallel, DBLP:conf/spaa/ShunDB14, DBLP:conf/wsdm/StergiouRT18}.
The computational power of current platforms and the characteristics of the problem graphs may partially explain such good performance.
Current parallel computing platforms such as MapReduce, Hadoop, Spark, and others have capabilities significantly beyond those modeled by the PRAM \cite{DBLP:journals/cacm/DeanG08}.
A more powerful model, the MPC (massively parallel computing) model \cite{DBLP:journals/jacm/BeameKS17} is intended to capture these capabilities.
In this model, each processor can have $\text{poly}(n)$ (typically sublinear in $n$) private memory, and the local computational power is unbounded.
A PRAM algorithm can usually be simulated on an MPC with asymptotically the same round complexity,
and it is widely believed that the MPC model admits algorithms faster than the PRAM model \cite{DBLP:conf/soda/KarloffSV10, DBLP:conf/isaac/GoodrichSZ11, DBLP:conf/podc/AssadiSW19}.
On the MPC model, and indeed on the weaker COMBINING CRCW PRAM model, there are very simple, practical algorithms that run in $O(\log n)$ time \cite{liu_tarjan}.
Furthermore, many graphs in applications have components of small diameter, perhaps poly-logarithmic in $n$.

These observations lead to the question of whether one can find connected components faster on graphs of small diameter, perhaps by exploiting the power of the MPC model.
Andoni et al. \cite{DBLP:conf/focs/Andoni} answered this question ``yes'' by giving an MPC algorithm that finds connected components in $O(\log d \log\log_{m / n} n)$ time, where $d$ is the largest diameter of a component.
Very recently, this time bound was improved to $O(\log d + \log\log_{m / n} n)$ by Behnezhad et al. \cite{DBLP:conf/focs/BehnezhadDELM19}.
Both of these algorithms are complicated and use the extra power of the MPC model, in particular, the ability to sort and compute prefix sums in $O(1)$ communication rounds.
These operations require $\Omega(\log n / \log\log n)$ time on a CRCW PRAM with $\text{poly}(n)$ processors \cite{DBLP:journals/jacm/BeameH89}.\footnote{Behnezhad et al. also consider the \emph{multiprefix} CRCW PRAM, in which prefix sum (and other primitives) can be computed in $O(1)$ time and $O(m)$ work. A direct simulation of this model on a PRIORITY CRCW PRAM or weaker model would suffer an $\Omega(\log n / \log\log n)$ factor in both time and work, compared to our result.}
These results left open the following fundamental problem in theory:\footnote{A repeated matrix squaring of the adjacency matrix computes the connected components in $O(\log d)$ time on a CRCW PRAM, but this is far from work-efficient -- the currently best work is  $O(n^{2.373})$ \cite{DBLP:conf/issac/Gall14a}.}

\begin{quote}
    \emph{Is it possible to break the $\log n$ time barrier for connected components of small-diameter graphs on a PRAM (without the additional power of an MPC)?}
\end{quote}

In this paper we give a positive answer by presenting an ARBITRARY CRCW PRAM algorithm that runs in $O(\log d + \log\log_{m/n} n)$ time, matching the round complexity of the MPC algorithm by Behnezhad et al. \cite{DBLP:conf/focs/BehnezhadDELM19}.
Our algorithm uses $O(m)$ processors and space, and thus is space-optimal and nearly work-efficient.
In contrast to the MPC algorithm, which uses several powerful primitives, we use only hashing and other simple data structures.
Our hashing-based approach also applies to the work of Andoni et al. \cite{DBLP:conf/focs/Andoni}, giving much simpler algorithms for connected components and spanning forest, which should be preferable in practice.
While the MPC model ignores the total work, our result on the more fine-grained PRAM model captures the inherent complexities of the problems.

\subsection{Computation Models and Main Results}

Our main model of computation is the ARBITRARY CRCW PRAM \cite{DBLP:journals/jal/Vishkin83}. %, which is a weaker model than the COMBINING CRCW PRAM \cite{DBLP:books/daglib}.
It consists of a set of processors , each of which has a constant number of cells (words) as the private memory, and a large common memory.
The processors run synchronously.
In one step, a processor can read a cell in common memory, write to a cell in common memory, or do a constant amount of local computation.
Any number of processors can read from or write to the same common memory cell concurrently.
If more than one processor writes to the same memory cell at the same time, an arbitrary one succeeds.
%Each cell of the memory holds $O(\log n)$ bits.

Our main results are stated below:

\begin{theorem}[Connected Components]\label{main1}
    There is an ARBITRARY CRCW PRAM algorithm using $O(m)$ processors that computes the connected components of any given graph.
    With probability $1 - 1 / \text{poly}((m \log n) / n)$, it runs in $O(\log d \log \log_{m/n} n)$ time.
\end{theorem}

The algorithm of Theorem~\ref{main1} can be extended to computing a spanning forest (a set of spanning trees of the components) with the same asymptotic efficiency:

\begin{theorem}[Spanning Forest]\label{main2}
    There is an ARBITRARY CRCW PRAM algorithm using $O(m)$ processors that computes the spanning forest of any given graph.
    With probability $1 - 1 / \text{poly}((m \log n) / n)$, it runs in $O(\log d \log \log_{m / n} n)$ time.
\end{theorem}

Using the above algorithms as bases, we provide a faster connected components algorithm that is nearly optimal (up to an additive factor of at most $O(\log\log n)$) due to a conditional lower bound of $\Omega(\log d)$ \cite{DBLP:journals/jacm/RoughgardenVW18, DBLP:conf/focs/BehnezhadDELM19}.

\begin{theorem}[Faster Connected Components]\label{main3}
    There is an ARBITRARY CRCW PRAM algorithm using $O(m)$ processors that computes the connected components of any given graph.
    With probability $1 - 1 / \text{poly}((m \log n) / n)$, it runs in $O(\log d + \log \log_{m / n} n)$ time.
\end{theorem}

For a dense graph with $m = n^{1 + \Omega(1)}$, the algorithms in all three theorems run in $O(\log d)$ time with probability $1 - 1 / \text{poly}(n)$;
if $d = \log_{m / n}^{\Omega(1)} n$, the algorithm in Theorem~\ref{main3} runs in $O(\log d)$ time.\footnote{Without the assumption that $m \ge 2n$ for simplification, one can replace the $m$ with $2(m + n)$ in all our statements by creating self-loops in the graph.}

Note that the time bound in Theorem~\ref{main1} is dominated by the one in Theorem~\ref{main3}.
We include Theorem~\ref{main1} here in pursuit of simpler proofs of Theorem~\ref{main2} and Theorem~\ref{main3}.

\subsection{Related Work and Technical Overview}\label{sec_overview}

In this section, we give an overview of the techniques in our algorithms,
highlighting the challenges in the PRAM model and the main novelty in our work compared to that of Andoni et al. \cite{DBLP:conf/focs/Andoni} and
Behnezhad et al. \cite{DBLP:conf/focs/BehnezhadDELM19}.

\subsubsection{Related Work}

Andoni et al.~\cite{DBLP:conf/focs/Andoni} observed that if every vertex in the graph has a degree of at least $b = m/n$, then one can choose each vertex as a \emph{leader} with probability $\Theta(\log(n)/b)$ to make sure that with high probability, every non-leader vertex has at least $1$ leader neighbor.
As a result, the number of vertices in the contracted graph is the number of leaders, which is $\tilde{O}(n/b)$ in expectation, leading to double-exponential progress, since we have enough space to make each remaining vertex have degree $\tilde{\Omega}(m / (n/b)) = \tilde{\Omega}(b^2)$ in the next round and $\tilde{\Omega}(b^{2^i})$ after $i$ rounds.\footnote{We use $\tilde{O}$ to hide $\text{polylog}(n)$ factors.}
The process of adding edges is called \emph{expansion}, as it expands the neighbor sets.
It can be implemented to run in $O(\log d)$ time.
This gives an $O(\log d \log\log_{m/n} n)$ running time.

Behnezhad et al. improved the multiplicative $\log\log_{m/n} n$ factor to an additive factor by streamlining the expansion procedure and double-exponential progress when increasing the space per vertex \cite{DBLP:conf/focs/BehnezhadDELM19}.
Instead of increasing the degree of each vertex \emph{uniformly} to at least $b$ in each round, they allow vertices to have different space budgets, so that the degree lower bound varies on the vertices.
They define the \emph{level} $\ell(v)$ of a vertex $v$ to control the \emph{budget} of $v$ (space owned by $v$):
initially each vertex is at level $1$ with budget $m/n$.
Levels increase over rounds.
Each $v$ is assigned a budget of $b(v) = (m/n)^{c^{\ell(v)}}$ for some fixed constant $c > 1$.
The maximal level $L \coloneqq \log_c\log_{m/n} n$ is such that a vertex at level $L$ must have enough space to find all vertices in its component.
Based on this idea, they design an MPC algorithm that maintains the following invariant:
\begin{observation}[\cite{DBLP:conf/focs/BehnezhadDELM19}] \label{ob_soheil1}
    With high probability, for any two vertices $u$ and $v$ at distance $2$, after $4$ rounds, if $\ell(u)$ does not increase then their distance decreases to $1$;
    moreover, the \emph{skipped} vertex $w$ originally between $u$ and $v$ satisfies $\ell(w) \le \ell(u)$.\footnote{For simplicity, we ignore the issue of changing the graph and corresponding vertices for now.}
\end{observation}

Behnezhad et al. proved an $O(\log d + \log\log_{m/n} n)$ time bound by
a potential-based argument on an arbitrary fixed shortest path $P_1$ in the original graph:
in round $1$ put $1$ coin on each vertex of $P_1$ (thus at most $d+1$ coins in total);
for the purpose of analysis only, when inductively constructing $P_{i+1}$ in round $i$, every skipped vertex on $P_i$ is removed and passes its coins evenly to its successor and predecessor (if exist) on $P_i$. They claimed that any vertex $v$ still on $P_i$ in round $i$ has at least ${1.1}^{i - \ell(v)}$ coins based on the following:
\begin{observation}[Claim 3.12 in \cite{DBLP:conf/focs/BehnezhadDELM19}] \label{ob_soheil2}
    For any path $P_i$ in round $i$ corresponding to an original shortest path, its first and last vertices are on $P_j$ in round $j$ for any $j > i$;
    moreover, if a vertex on $P_i$ does not increase level in $4$ rounds, then either its predecessor or successor on $P_i$ is skipped.
\end{observation}
The first statement is to maintain the connectivity for every pair of vertices in the original graph, and the second statement is to guarantee that each vertex obtains enough coins in the next round.
(Observation~\ref{ob_soheil2} is seemingly obvious from Observation~\ref{ob_soheil1}, however there is a subtle issue overlooked in \cite{DBLP:conf/focs/BehnezhadDELM19} that invalidates the statement.
Their bound is still valid without changing the algorithm by another potential-based argument, which shall be detailed later in this section.)

Since the maximal level is $L$, by the above claim,
a vertex on such a path with length more than $1$ in round $R \coloneqq 8 \log d +  L$ would have at least $1.1^{8 \log d} > d+1$ coins, a contradiction, giving the desired time bound.

%Before addressing the challenge and insight in our PRAM algorithm,
Let us first show a counter-example for Observation~\ref{ob_soheil2} (not for their time bound).
Let the original path $P_1$ be $(v_1, v_2, \dots, v_s)$.
We want to show that the distance between $v_1$ and $v_s$ is at most $1$ after $R$ rounds, so neither $v_1$ nor $v_s$ can be skipped during the path constructions over rounds.
Suppose $v_1$ does not increase level in $4$ rounds, then for $v_1$ to obtain enough coins to satisfy the claim, $v_2$ must be skipped.
We call an ordered pair $(v_i, v_{i+1})$ of consecutive vertices on the path \emph{frozen} if $v_i$ is kept and $v_{i+1}$ is skipped in $4$ rounds.
So $(v_1, v_2)$ is frozen.
Let $v_3$ (any vertex after $v_2$ suffices) be the \emph{only} vertex on $P_1$ directly connecting to $v_1$ after skipping $v_2$ in $4$ rounds.
(Observation~\ref{ob_soheil1} guarantees that a vertex directly connecting to $v_1$ must exist but cannot guarantee that there is more than $1$ such vertex on $P_1$.)
Note that $v_3$ \emph{cannot} be skipped, otherwise $v_1$ is isolated from $P_1$ and thus cannot connect to $v_s$.
Now consider two cases.
If $v_4$ is skipped, then pair $(v_3, v_4)$ is frozen.
If $v_4$ is not skipped, then assume $v_4$ does not increase level in $4$ rounds.
For $v_4$ to obtain enough coins, $v_5$ must be skipped because $v_3$ cannot be skipped to pass coins to $v_4$, so the pair $(v_4, v_5)$ is frozen.
Observe that from frozen pair $(v_1, v_2)$, in either case we get another frozen pair, which propagates inductively to the end of $P_1$.
If we happen to have a frozen pair $(v_{s-1}, v_s)$, then $v_s$ must be skipped and isolated from $v_1$, a contradiction.

%Without changing the algorithm,
Here is a fix to the above issue (formally stated in Lemma~\ref{lem_dr2}).
Note that only the last vertex $v$ on the path can violate the claim that there are at least ${1.1}^{i - \ell(v)}$ coins on $v$ in round $i$.
Assuming the claim holds for all vertices on the current path, one can always \emph{drop} (to distinguish from skip) $v_s$ and pass its coins to $v_{s-1}$ (which must be kept) if they are a frozen pair.
So $v_{s-1}$, the new last vertex after $4$ rounds, obtains enough coins by the induction hypothesis and the second part of Observation~\ref{ob_soheil1}.
By the same argument, the resulting path after $R$ rounds has length at most $1$.
Observe that we dropped $O(R)$ vertices consecutively located at the end of the path, so the concatenated path connecting $v_1$ and the original $v_s$ has length $O(R)$.
Now applying Observation~\ref{ob_soheil1} to $v_1$, we get that in $4$ rounds, either its level increases by $1$ or its successor is skipped.
Therefore, in $O(R + L)$ rounds, there is no successor of $v_1$ to be skipped and the graph has diameter at most $1$ by a union bound over all the (original) shortest paths.

\subsubsection{Our Contributions}\label{sec_ov_contribution}
Now we introduce the new algorithmic ideas in our PRAM algorithm with a matching time bound.

%Recall that in Observation~\ref{ob_cc_alg}, we do not care about the dependency among the events that a vertex owning a block or obtaining enough neighbors after the expansion, because the progress of the execution is measured by the number of remaining vertices, which works with good probability even if everything holds in expectation.
%However in the faster connected components algorithm, we require high-probability results for all vertices since we need to do a union bound over all the $O(n^2)$ shortest paths.

The first challenge comes from processor allocation. %: the hashing-based method no longer suffices.
To allocate different-sized blocks of processors to vertices in each round, there is actually an existing tool called \emph{approximate compaction}, which maps the $k$ distinguished elements in a length-$n$ array one-to-one to an array of length $O(k)$ with high probability \cite{DBLP:books/daglib}.
(The vertices to be assigned blocks are \emph{distinguished} and their names are indices in the old array; after indexing them in the new array, one can assign them predetermined blocks.)
However, the current fastest (and work-optimal) approximate compaction algorithm takes $O(\log^* n)$ time, introducing a multiplicative factor \cite{DBLP:conf/focs/GilMV91}.
To avoid this, our algorithm first reduces the number of vertices to $n / \text{polylog}(n)$ in $O(\log\log_{m/n} n)$ time, then uses approximate compaction to rename the remaining vertices by an integer in $[n / \text{polylog}(n)]$ in $O(\log^* n)$ time.
After this, each cell of the array to be compacted owns $\text{polylog}(n)$ processors, and each subsequent compaction (thus processor allocation) can be done in $O(1)$ time \cite{DBLP:conf/focs/Goodrich91}.

The second challenge is much more serious: it is required by the union bound that any vertex $u$ must connect to \emph{all} vertices within distance $2$ with high probability if $u$ does not increase in level.
Behnezhad et al. achieve this goal by an algorithm based on constant-time sorting and prefix sum, which require $\Omega(\log n / \log\log n)$ time on an ARBITRARY CRCW PRAM with $\text{poly}(n)$ processors \cite{DBLP:journals/jacm/BeameH89}.
Our solution is based on hashing: %, but leveraged in a novel way:
to expand a vertex $u$, hash all vertices within distance $2$ from $u$ to a hash table owned by $u$;
if there is a collision, %in hashing vertices within distance $2$ from $u$ into the block of $u$,
increase the level of $u$.\footnote{Hashing also naturally removes the duplicate neighbors to get the desired space bound -- a goal achieved by sorting in \cite{DBLP:conf/focs/Andoni, DBLP:conf/focs/BehnezhadDELM19}.}
As a result, we are able to prove the following result, which is stronger than Observation~\ref{ob_soheil1} as it holds deterministically and uses only $1$ round:
\begin{observation}[formally stated in Lemma~\ref{lem_2hops}] \label{ob_faster_cc}
    For any two vertices $u$ and $v$ at distance $2$, if $\ell(u)$ does not increase then their distance decreases to $1$ in the next round;
    moreover, any vertex originally between $u$ and $v$ has level at most $\ell(u)$.
\end{observation}

The idea of increasing the level immediately after seeing a collision gives a much cleaner proof of Observation~\ref{ob_faster_cc}, but might be problematic in bounding the total space/number of processors:
a vertex with many vertices within distance $2$ can incur a collision and level increase very often.
We circumvent this issue by increasing the level of each budget-$b$ vertex with probability $\tilde{\Theta}(b^{-\delta})$
%(similar to leader selection in \cite{DBLP:conf/focs/Andoni, DBLP:conf/focs/BehnezhadDELM19})
before hashing.
Then a vertex $v$ with at least $b^{\delta}$ vertices within distance $2$ would see a level increase with high probability;
if the level does not increase, there should be at most $b^{\delta}$ vertices within distance $2$, thus there is a collision when expanding $v$ with probability $1 / \text{poly}(b)$.
As a result, the probability of level increase is $\tilde{\Theta}(b^{-\delta}) + 1 / \text{poly}(b) \le b^{-c}$ for some constant $c > 0$, and we can assign a budget of $b^{1 + \Omega(c)}$ to a vertex with increased level, leading to double-exponential progress.
As a result, the total space is $O(m)$ with good probability since the union bound is over all $\text{polylog}(n)$ different levels and rounds, instead of $O(n^2)$ shortest paths.

Suitable combination of the ideas above yields a PRAM algorithm that reduces the diameter of the graph to at most $1$ in $O(R) = O(\log d + \log\log_{m/n} n)$ time, with one \emph{flexibility}:
the relationship between the level $\ell(v)$ and budget $b(v)$ of vertex $v$ in our algorithm is \emph{not} strictly $b(v) = (m/n)^{c^{\ell(v)}}$ for some fixed constant $c > 1$ as in Behnezhad et al. \cite{DBLP:conf/focs/BehnezhadDELM19};
instead, we allow vertices with the same level to have \emph{two} different budgets.
We show that such flexibility still maintains the key invariant of our algorithm (without influencing the asymptotic space bound):
if a vertex is not a root in a tree in the labeled digraph, then its level must be strictly lower than the level of its parent (formally stated in Lemma~\ref{lem_non_root}).\footnote{Our algorithm adopts the framework of \emph{labeled digraph} (or \emph{parent graph}) for computing and representing components, which is standard in PRAM literatures, see \S{\ref{pre}}.}
Using hashing and a proper parent-update method, our algorithm does not need to compute the number of neighbors with a certain level for \emph{each} vertex, which is required in \cite{DBLP:conf/focs/Andoni, DBLP:conf/focs/BehnezhadDELM19} and solved by constant-time sorting and prefix sum on an MPC.
If this were done by a direct application of (constant-time) \emph{approximate counting} (cf. \cite{DBLP:conf/dimacs/Ajtai90}) on each vertex, then each round would take $\Omega(k)$ time where $k$ is the maximal degree of any vertex, so our new ideas are essential to obtain the desired time bound.

Finally, we note that while it is straightforward to halt when the graph has diameter at most $1$ in the MPC algorithm,
it is not correct to halt (nor easy to determine) in this case due to the different nature of our PRAM algorithm.
After the diameter reaches $O(1)$, to correctly compute components and halt the algorithm, we borrow an idea from \cite{liu_tarjan} to flatten all trees in the labeled digraph in $O(R)$ time, then apply our \emph{slower} connected components algorithm (cf. Theorem~\ref{main1}) to output the correct components in $O(\log\log_{m/n} n)$ time, which is $O(R)$ total running time.

\begin{comment}
\subsection{Roadmap}

The rest of this paper is organized as follows.
In \S{\ref{pre}}, we introduce our framework of the algorithms and some essential building blocks.
%After that in \S{\ref{sec_cc_alg}}, we give our connected components algorithm and its analysis to prove Theorem~\ref{main1}.
%Based on \S{\ref{sec_cc_alg}}, the spanning forest algorithm is presented in \S{\ref{sec_sf_alg}}, as well as a proof of Theorem~\ref{main2}.
%Finally,
Our faster connected components algorithm with running time $O(\log d + \log\log_{m/n} n)$ is presented in \S{\ref{sec_faster_cc_alg}}.
We include the analysis of its correctness, running time, number of processors, success probability, and implementation in each subsection of \S{\ref{sec_faster_cc_alg}} respectively.
The full proofs of Theorem~\ref{main1}, Theorem~\ref{main2}, and Theorem~\ref{main3} are deferred to the appendix.

%We conclude this paper with some remarks and open problems in \S{\ref{last_section}}.

\end{comment}

\section{Preliminaries}\label{pre}

\subsection{Framework}

%We use the \emph{leader contraction} approach for finding components concurrently.
We formulate the problem of computing connected components concurrently as follows: label each vertex $v$ with a unique vertex $v.p$ in its component.
Such a labeling gives a constant-time test for whether two vertices $v$ and $w$ are in the same component: they are if and only if $v.p = w.p$.
We begin with every vertex self-labeled ($v.p = v$) and successively update labels until there is exactly one label per component.

The labels define a directed graph (\emph{labeled digraph}) with arcs $(v, v.p)$, where $v.p$ is the \emph{parent} of $v$.
We maintain the invariant that the only cycles in the labeled digraph are self-loops (arcs of the form $(v, v)$).
Then this digraph consists of a set of rooted trees, with $v$ a root if and only if $v = v.p$. %and $v.p$ the parent of $v$ if $v \neq v.p$ (hence our use of $v.p$, for parent, to denote the label of $v$).
Some authors call the root of a tree the \emph{leader} of all its vertices.   We know of only one algorithm in the literature that creates non-trivial cycles, that of Johnson and Metaxis \cite{DBLP:journals/jal/JohnsonM95}.  Acyclicity implies that when the parent of a root $v$ changes, the new parent of $v$ is not in the tree rooted at $v$ (for any order of the concurrent parent changes).
We call a tree \emph{flat} if the root is the parent of every vertex in the tree.
Some authors call flat trees \emph{stars}.

In our connected components and spanning forest algorithms (cf. \S{\ref{sec_cc_alg}} and \S{\ref{sec_sf_alg}} in the appendix),
we maintain the additional invariant that if the parent of a non-root $v$ changes, its new parent is in the same tree as $v$ (for any order of the parent changes).
This invariant implies that the partition of vertices among trees changes only by set union; that is, no parent change moves a proper subtree to another tree.
We call this property \emph{monotonicity}.
Most of the algorithms in the literature that have a correct efficiency analysis are monotone.
Liu and Tarjan \cite{liu_tarjan} analyze some non-monotone algorithms.
In our faster connected components algorithm (cf. \S{\ref{sec_faster_cc_alg}}), only the preprocessing and postprocessing stages are monotone, which means the execution between these two stages can move subtrees between different trees in the labeled digraph.

\subsection{Building Blocks} \label{sec_bb}

Our algorithms use three standard and one not-so-standard building blocks, which link (sub)trees, flatten trees, alter edges, and add edges, respectively.
(Classic PRAM algorithms develop many techniques to make the graph sparser, e.g., in \cite{DBLP:journals/siamcomp/Gazit91, DBLP:journals/jcss/HalperinZ96, halperin2001optimal}, not denser by adding edges.)

We treat each edge $\{v, w\}$ as a pair of oppositely directed arcs $(v, w)$ and $(w, v)$.
A \emph{direct link} applies to a graph arc $(v, w)$ such that $v$ is a root and $w$ is not in the tree rooted at $v$; it makes $w$ the parent of $v$.
A \emph{parent link} applies to a graph arc $(v, w)$ and makes $w.p$ the parent of $v$; note that $v$ and $w.p$ are not necessarily roots.
Concurrent direct links maintain monotonicity while concurrent parent links do not.
We add additional constraints to prevent the creation of a cycle in both cases.
Specifically, in the case of parent links, if a vertex is not a root in a tree in the labeled digraph, then its level must be strictly lower than the level of its parent (formally stated in Lemma~\ref{lem_non_root}).
%Other kinds of links are possible and useful, for example making $w.p$ the parent of $v$, but direct and parent linking (and, indeed, just one or the other) suffice for us.

Concurrent links can produce trees of arbitrary heights.
To reduce the tree heights, we use the \emph{shortcut} operation: for each $v$ do $v.p \coloneqq v.p.p$.
One shortcut roughly halves the heights of all trees; $O(\log n)$ shortcuts make all trees flat.
Hirschberg et al. \cite{DBLP:journals/cacm/HirschbergCS79}
introduced shortcutting in their connected components algorithm; it is closely related to the \emph{compress} step in tree contraction \cite{DBLP:conf/focs/MillerR85} and to \emph{path splitting} in disjoint-set union \cite{DBLP:journals/jacm/TarjanL84}.

%Suitable combinations of (direct or parent) links and shortcuts suffice to efficiently compute components.
%Both direct links and parent links update the parent of $v$, which needs an operation that changes graph edges to correctly compute connected components.
Our third operation changes graph edges.
To \emph{alter} $\{v, w\}$, we replace it by $\{v.p, w.p\}$. % unless $v.p = w.p$, in which case we delete the edge.
Links, shortcuts, and edge alterations suffice to efficiently compute components.
Liu and Tarjan \cite{liu_tarjan} analyze simple algorithms that use combinations of our first three building blocks.

To obtain a good bound for small-diameter graphs, we need a fourth operation that adds edges.
We \emph{expand} a vertex $u$ by adding an edge $\{u, w\}$ for a neighbor $v$ of $u$ and a neighbor $w$ of $v$.
%Expansion can create multiple edges; it is up to the algorithm to delete extra copies or not.
The key idea for implementing expansion is hashing, which is presented below.

Suppose each vertex owns a block of $K^2$ processors.
For each processor in a block, we index it by a pair $(p, q) \in [K] \times [K]$.
%If the index of a block is hashed by only one vertex, then we use this block of processors to expand the vertex.
%In the above processors allocation procedure, if we divide processors into a suitable number of blocks, then there are sufficient number of vertices with each associated with an exclusive block with good probability.
For each vertex $u$, we maintain a size-$K$ table $H(u)$.
We choose a random hash function $h: [n] \rightarrow [K]$.
At the beginning of an expansion, for each graph arc $(u, v)$, we write vertex $v$ into the $h(v)$-th cell of $H(u)$.
Then we can expand $u$ as follows:
each processor $(p, q)$ reads vertex $v$ from the $p$-th cell of $H(u)$, reads vertex $w$ from the $q$-th cell of $H(v)$, and writes vertex $w$ into the $h(w)$-th cell of $H(u)$.
For each $w \in H(u)$ after the expansion, $\{u, w\}$ is considered an \emph{added} edge in the graph and is treated the same as any other edge.

The key difference between our hashing-based expansion and that in the MPC algorithms is that a vertex $w$ within distance $2$ from $u$ might not be in $H(u)$ after the expansion due to a collision,
so crucial to our analysis is the way to handle collisions.
All hash functions in this paper are pairwise independent, so each processor doing hashing in each round only needs to read two words, which uses $O(1)$ private memory and time.

\section{An $O(\log d \log\log_{m/n} n)$-time Connected Components Algorithm}\label{sec_cc_alg}

In this section we present our connected components algorithm.
For simplicity in presentation, we also consider the COMBINING CRCW PRAM \cite{DBLP:books/daglib}, whose computational power is between the ARBITRARY CRCW PRAM and MPC.
This model is the same as the ARBITRARY CRCW PRAM, except that if several processors write to the same memory cell at the same time, the resulting value is a specified symmetric function (e.g., the sum or min) of the individually written values.

We begin with a simple randomized algorithm proposed by Reif \cite{reif1984optimal} but adapted in our framework, which is called Vanilla algorithm.
Our main algorithm will use the same framework with a simple preprocessing method to make the graph denser and an elaborated expansion method to add edges before the direct links.
%For a simpler presentation,
We implement the algorithm on a COMBINING CRCW PRAM and then generalize it to run on an ARBITRARY CRCW PRAM in \S{\ref{remove_assumption}}.

%If $u$ is a vertex such that $u.p \neq u$, $u$ is a \emph{child} of $u.p$.

%We assume the graph is connected and contains at least two vertices, which is without loss of generality since each algorithm operates concurrently and independently on each component.
\subsection{Vanilla Algorithm}

In each iteration of Vanilla algorithm (see below), some roots of trees are selected to be leaders, which becomes the parents of non-leaders after the \textsc{link}.
A vertex $u$ is called a \emph{leader} if $u.l = 1$.
%, and a \emph{non-leader} if $u.l = 0$.

\begin{framed}
\noindent Vanilla algorithm: repeat \{\textsc{random-vote}; \textsc{link}; \textsc{shortcut}; \textsc{alter}\} until no edge exists other than loops.

\textsc{random-vote}: for each vertex $u$: set $u.l \coloneqq 1$ with probability $1/2$, and $0$ otherwise.

\textsc{link}: for each graph arc $(v, w)$: if $v.l = 0$ and $w.l = 1$ then update $v.p$ to $w$.

\textsc{shortcut}: for each vertex $u$: update $u.p$ to $u.p.p$.

\textsc{alter}: for each edge $e = \{v, w\}$: replace it by $\{v.p, w.p\}$. %, \emph{if this creates a loop then delete it}.
%\footnote{In the spanning forest algorithm, the algorithm also needs to store the original edge for the output.}
\end{framed}

It is easy to see that
Vanilla algorithm uses $O(m)$ processors and can run on an ARBITRARY CRCW PRAM.
We call an iteration of the repeat loop in the algorithm a \emph{phase}.
Clearly each phase takes $O(1)$ time.
We obtain the following results:
\begin{definition}\label{def_ongoing}
    A vertex is \emph{ongoing} if it is a root but not the only root in its component, otherwise it is \emph{finished}.
\end{definition}

\begin{lemma}\label{all_tree_flat}
    At the beginning of each phase, each tree is flat and  a vertex is ongoing if and only if it is incident with a non-loop edge.
\end{lemma}
\begin{proof}
    The proof is by induction on phases.
    At the beginning, each vertex is in a single-vertex tree and the edges between trees are in the original graph, so the lemma holds.
    Suppose it holds for phase $k-1$.
    After the \textsc{link} in phase $k$, each tree has height at most two since only non-leader root can update its parent to a leader root.
    The following \textsc{shortcut} makes the tree flat, then the \textsc{alter} moves all the edges to the roots.
    If a root is not the only root in its component, there must be an edge between it and another vertex not in its tree.
\end{proof}

\begin{lemma}\label{van_prob1}
    Given a vertex $u$, after $k$ phases of Vanilla algorithm, $u$ is ongoing with probability at most $(3 / 4)^k$.
\end{lemma}
\begin{proof}
    We prove the lemma by an induction on $k$.
    The lemma is true for $k = 0$.
    Suppose it is true for $k-1$.
    Observe that a non-root can never again be a root.
    For vertex $u$ to be ongoing after $k$ phases, it must be ongoing after $k-1$ phases.
    By the induction hypothesis this is true with probability at most $(3/4)^{k-1}$.
    Furthermore, by Lemma~\ref{all_tree_flat}, if this is true, there must be an edge $\{u,v\}$ such that $v$ is ongoing.
    With probability $1/4$, $u.l = 0$ and $v.l = 1$, then $u$ is finished after phase $k$.
    It follows that the probability that, after phase $k$, $u$ is still ongoing is at most $(3/4)^k$.
\end{proof}

By Lemma~\ref{van_prob1}, the following corollary is immediate by linearity of expectation and Markov's inequality:
\begin{corollary}\label{van_prob}
    After $k$ phases of Vanilla algorithm, the number of ongoing vertices is at most $(7/8)^{k} n$ with probability at least $1 - (6/7)^k$.
\end{corollary}

By Lemma~\ref{all_tree_flat}, Corollary~\ref{van_prob}, and monotonicity, we have that Vanilla algorithm outputs the connected components in $O(\log n)$ time with high probability.

\subsection{Algorithmic Framework} \label{sec_cc_alg_framwork}

%(Need to introduce the algorithm of Andoni et al. at some place before,  \cite{DBLP:conf/focs/Andoni}.)

In this section we present the algorithmic framework for our connected components algorithm.

%The current graph contains two types of edges: (\romannumeral1) the edges corresponding to edge processors that are updated by \textsc{alter}, and (\romannumeral2) the added edges by \textsc{expand} with
For any vertex $v$ in the current graph, a vertex within distance $1$ from $v$ in the current graph (which contains the (altered) original edges and the added edges) is called a \emph{neighbor} of $v$. The set of neighbors of every vertex is maintained during the algorithm (see the implementation of the \textsc{expand}).
%For each ongoing vertex $u$, we assign every vertex $v$ in its neighbor set a process, then $v$ knows about $u$ and can send a message to $u$.
%For any root $u$, if a vertex $v$ in the neighbor set of $u$ has its own processor apart from the vertex processor associated to $v$, then $v$ knows about $u$ and can update the parent of $u$ in $O(1)$ time.
%If for every root, the number of its neighbors is properly upper bounded, one can assign every vertex in the neighbor set a processor and the total number of processors is still $O(m)$.
\begin{framed}
\noindent Connected Components algorithm: \textsc{prepare}; repeat \{\textsc{expand}; \textsc{vote}; \textsc{link}; \textsc{shortcut}; \textsc{alter}\} until no edge exists other than loops.

\textsc{prepare}: if $m / n \le \log^c n$ for given constant $c$ then run $c \log_{8/7}\log n$ phases of Vanilla algorithm.

\textsc{expand}: for each ongoing $u$: expand the neighbor set of $u$ according to some rule.

\textsc{vote}: for each ongoing $u$: set $u.l$ according to some rule in $O(1)$ time.

%\textsc{fix}: \emph{for each root $u$: if there is no $v$ in the neighbor set of $u$ such that $v.l = 1$ then set $u.l = 1$}.

%\textsc{link}: \emph{for each ongoing $v$}: \emph{for each $w$ in the neighbor set of $v$}: \emph{if $v.l = 0$ and $w.l = 1$ then update $v.p$ to $w$}.
\textsc{link}: for each ongoing $v$: for each $w$ in the neighbor set of $v$: if $v.l = 0$ and $w.l = 1$ then update $v.p$ to $w$.
\end{framed}

The \textsc{shortcut} and \textsc{alter} are the same as those in Vanilla algorithm.
The \textsc{link} is also the same in the sense that in our algorithm the graph arc $(v, w)$ is added during the \textsc{expand} in the form of adding $w$ to the neighbor set of $v$.
Therefore, Lemma~\ref{all_tree_flat} also holds for this algorithm.

The details of the \textsc{expand} and \textsc{vote} will be presented in \S{\ref{exp_ana}} and \S{\ref{TRA}}, respectively.
We call an iteration of the repeat loop after the \textsc{prepare} a \emph{phase}.
By Lemma~\ref{all_tree_flat}, we can determine whether a vertex is ongoing by checking the existence of non-loop edges incident on it,
therefore in each phase, the \textsc{vote}, \textsc{link}, \textsc{shortcut}, and \textsc{alter} take $O(1)$ time.

%It also takes $O(1)$ time to determine whether a vertex $u$ is ongoing dues to Lemma~\ref{all_tree_flat}: if and only if there is an incident edge on $u$.

Let $\delta = m / n'$, where $n'$ is the number of ongoing vertices at the beginning of a phase.
Our goal in one phase is to reduce $n'$ by a factor of at least a positive constant power of $\delta$ with high probability with respect to $\delta$,
so we do a \textsc{prepare} before the main loop to obtain a large enough $\delta$ with good probability:
\begin{lemma}\label{prepare}
    After the \textsc{prepare}, if $m / n > \log^c n$, then $m / n' \ge \log^c n$;
    otherwise $m / n' \ge \log^c n$ with probability at least $1 - 1/\log^c n$.
\end{lemma}
\begin{proof}
    The first part is trivial since the \textsc{prepare} does nothing.
    By Corollary~\ref{van_prob}, after $c \log_{8/7}\log n$ phases, there are at most $n / \log^c n$ ongoing vertices with probability at least $1 - (6/7)^{c \log_{8/7}\log n} \ge 1 - 1/\log^c n$, and the lemma follows immediately from $m \ge n$.
\end{proof}

We will be focusing on the \textsc{expand}, \textsc{vote}, and \textsc{link}, so in each phase it suffices to only consider the induced graph on ongoing vertices with current edges.
If no ambiguity, we call this induced graph just the graph, call the current edge just the edge, and call an ongoing vertex just a vertex.

%Later in \S{\ref{rma}} we will show that any large enough constant $c$ suffices for proving the claimed result.

%By Lemma~\ref{prepare}, for some sparse graphs with $m / n \le \log^c n$, at the beginning we obtain an upper bound of $n'$ with high probability, however
In the following algorithms and analyses, we will use the following assumption for simplicity in the analyses.
\begin{assumption}\label{assumption1}
    The number of ongoing vertices $n'$ is known at the beginning of each phase.
\end{assumption}

This holds if running on a COMBINING CRCW PRAM with sum as the combining function to compute $n'$ in $O(1)$ time.
Later in \S{\ref{remove_assumption}} we will show how to remove Assumption~\ref{assumption1} to implement our algorithms on an ARBITRARY CRCW PRAM.

\subsection{The Expansion}\label{exp_ana}

In this section, we present the method \textsc{expand} and show that almost all vertices have a large enough neighbor set after the \textsc{expand} with good probability.
%For simplicity, all vertices discussed in this section are ongoing vertices.

\subsubsection{Setup}

\paragraph{Blocks.}
We shall use a pool of $m$ processors to do the \textsc{expand}.
We divide these into $m / \delta^{2/3}$ indexed \emph{blocks}, where each block contains $\delta^{2/3}$ indexed processors.
Since $n'$ and $\delta$ are known at the beginning of each phase (cf. Assumption~\ref{assumption1}), if a vertex is assigned to a block, then it is associated with $\delta^{2/3}$ (indexed) processors.
We map the $n'$ vertices to the blocks by a random hash function $h_B$.
Each vertex has a probability of being the only vertex mapped to a block, and if this happens then we say this vertex \emph{owns} a block.

\paragraph{Hashing.}
We use a hash table to implement the neighbor set of each vertex and set the size of the hash table as $\delta^{1/3}$, because we need $\delta^{1/3}$ processors for each cell in the table to do an expansion step  (see Step~(\ref{second_hash_step}) in the \textsc{expand}).
We use a random hash function $h_V$ to hash vertices into the hash tables.
Let $H(u)$ be the hash table of vertex $u$.
If no ambiguity, we also use $H(u)$ to denote the set of vertices stored in $H(v)$.
If $u$ does not own a block, we think that $H(u) = \emptyset$.

We present the method \textsc{expand} as follows:
\begin{framed}
\noindent \textsc{expand}:
\begin{enumerate}
    \item Each vertex is either \emph{live} or \emph{dormant} in a step. Mark every vertex as \emph{live} at the beginning.
    \item Map the vertices to blocks using $h_B$. Mark the vertices that do not own a block as \emph{dormant}. \label{first_dormant}
    \item For each graph arc $(v, w)$: if $v$ is live before Step~(\ref{first_hash_step}) then use $h_V$ to hash $v$ into $H(v)$ and $w$ into $H(v)$, else mark $w$ as \emph{dormant}. \label{first_hash_step}
    \item For each hashing done in Step~(\ref{first_hash_step}): if it causes a \emph{collision} (a cell is written by different values) in $H(u)$ then mark $u$ as \emph{dormant}. \label{second_dormant}
    \item Repeat the following until there is neither live vertex nor hash table getting a new entry: \label{expand_loop}
        \begin{enumerate}
            %\item For each live root $u$, if there exists a dormant root in $H(u)$ then mark $u$ as \emph{dormant}. \label{no_dormant_neighbor}
            \item For each vertex $u$: for each $v$ in $H(u)$: if $v$ is dormant before Step~(\ref{second_hash_step}) in this iteration then mark $u$ as \emph{dormant}, for each $w$ in $H(v)$: use $h_V$ to hash $w$ into $H(u)$.  \label{second_hash_step}
            \item For each hashing done in Step~(\ref{second_hash_step}): if it causes a collision in $H(u)$ then mark $u$ as \emph{dormant}. \label{third_dormant}
        \end{enumerate}
\end{enumerate}
\end{framed}
The first four steps and each iteration of Step~(\ref{expand_loop}) in the \textsc{expand} take $O(1)$ time.
We call an iteration of the repeat loop in Step~(\ref{expand_loop}) a \emph{round}.
We say a statement holds before round $0$ if it is true before Step~(\ref{first_hash_step}),
it holds in round $0$ if it is true after Step~(\ref{second_dormant}) and before Step~(\ref{expand_loop}), and it holds in round $i~(i > 0)$ if it is true just after $i$ iterations of the repeat loop in Step~(\ref{expand_loop}).

%\begin{lemma}\label{def_dormant}
%    A root is \emph{dormant} in a round if it does not own a block, or its hash table in that round contains a collision or a root that is dormant before that round.
%    A dormant root remains to be a dormant root until the next phase.
%\end{lemma}
%\begin{proof}
%    By the implementation of \textsc{expand} above.
%\end{proof}

%\subsection{The Expansion}

\paragraph{Additional notations.}
We use $\text{dist}(u, v)$ to denote the \emph{distance} between $u$ and $v$, which is the length of the shortest path from $u$ to $v$.
We use $B(u, \alpha) = \{v \in V \mid \text{dist}(u, v) \le \alpha \}$ to represent the set of vertices with distance at most $\alpha$ from $u$.
For any $j \ge 0$ and any  vertex $u$, let $H_j(u)$ be the hash table of $u$ in round $j$.

Consider a vertex $u$ that is dormant after the \textsc{expand}.
We call $u$ \emph{fully dormant} if $u$ is dormant before round $0$, i.e., $u$ does not own a block. Otherwise, we call $u$ \emph{half dormant}.
For a half dormant $u$, let $i \ge 0$ be the first round $u$ becomes dormant.
For $u$ that is live after the \textsc{expand}, let $i \ge 0$ be the first round that its hash table is the same as the table just before round $i$.
The following lemma shows that in this case, the table of $u$ in round $j < i$ contains exactly the vertices within distance $2^j$:

\begin{lemma}\label{radius_doubling}
    For any vertex $u$ that is not fully dormant, let $i$ be defined above, then it must be that $H_i(u) \subseteq B(u, 2^i)$.
    Furthermore, for any $j \in [0, i-1]$, $H_j(u) = B(u, 2^j)$.
\end{lemma}
\begin{proof}
    According to the update rule of $H(u),$ it is easy to show that for any integer $j \geq 0$, $H_j(u) \subseteq B(u, 2^j)$ holds by induction.
    Now we prove that for any $j\in[0,i-1]$, $H_j(u)=B(u,2^j)$.
    We claim that for any vertex $v$, if $v$ is not dormant in round $j$, then $H_j(v)=B(v,2^j)$. The base case is when $j=0$. In this case, $H_0(v)$ has no collision, so the claim holds. Suppose the claim holds for $j-1$, i.e., for any vertex $v'$ which is not dormant in round $j-1$, it has $H_{j-1}(v')=B(v',2^{j-1})$. Let $v$ be any vertex which is not dormant in round $j$. Then since there is no collision, $H_j(v)$ should be $\bigcup_{v'\in H_{j-1}(v) }H_{j-1}(v')=\bigcup_{v'\in B(v,2^{j-1})} B(v',2^{j-1})=B(v,2^j)$.
    Thus the claim is true, and it implies that for any $j\in[0,i-1]$, $H_j(u)=B(u,2^j)$.
\end{proof}

\begin{lemma}\label{expansion_rounds_bound}
    The \textsc{expand} takes $O(\log d)$ time.
\end{lemma}
\begin{proof}
    By an induction on phases, any path in the previous phase is replaced by a new path with each vertex on the old path replaced by its parent, so the diameter never increases.
    Since $u$ either is fully dormant or stops its expansion in round $i$, the lemma follows immediately from Lemma~\ref{radius_doubling}.
\end{proof}

\subsubsection{Neighbor Set Size Lower Bound}

We want to show that the table of $u$ in round $i$ contains enough neighbors, but $u$ becomes dormant in round $i$ possibly dues to propagations from another vertex in the table of $u$ that is dormant in round $i-1$, which does not guarantee the existence of collisions in the table of $u$ (which implies large size of the table with good probability).
We overcome this issue by identifying the maximal-radius ball around $u$ with no collision nor fully dormant vertex, whose size serves as a size lower bound of the table in round $i$.
\begin{definition}\label{def_r}
    For any vertex $u$ that is dormant after the \textsc{expand}, let $r$ be the minimal integer such that
    there is no collision nor fully dormant vertex in $B(u, r-1)$.
\end{definition}

\begin{lemma}\label{r_range}
    If $u$ is fully dormant then $r = 0$.
    If $u$ is half dormant then
    $ 2^{i-1} < r \le 2^i $.
\end{lemma}
\begin{proof}
    If $u$ is fully dormant, then $r = 0$ since $B(u, 0) = \{u\}$.
    Suppose $u$ is half dormant.
    We prove the lemma by induction on $i$.
    The lemma holds for $i = 0$ because $r > 0$ and if $r \ge 2$ then $u$ cannot be dormant in Step~(\ref{first_hash_step}) nor Step~(\ref{second_dormant}).
    The lemma also holds for $i = 1$ because if $r = 1$ then $u$ becomes dormant in Step~(\ref{first_hash_step}) or Step~(\ref{second_dormant}) and if $r \ge 3$ then $u$ cannot be dormant in round $1$.

    Suppose $i \ge 2$.
    Assume $r \le 2^{i-1}$ and let $v \in B(u, r)$ be a fully dormant vertex or a vertex that causes a collision in $B(u, r)$.
    Assume $u$ becomes dormant after round $i-1$.
    By Lemma~\ref{radius_doubling}, we know that $H_{i-1}(u)=B(u,2^{i-1})$. Since $r\leq 2^{i-1}$, there is no collision in $B(u,r)$ using $h_V$. Thus, there is a fully dormant vertex $v$ in $B(u,r)\subseteq H_{i-1}(u)$. Consider the first round $j\leq i-1$ that $v$ is added into $H(u)$. If $j=0,$ then $u$ is marked as dormant in round $0$ by Step~(\ref{first_hash_step}).
    If $j>0$, then in round $j$, there is a vertex $v'$ in $H_{j-1}(u)$ such that $v\in H_{j-1}(v')$, and $v$ is added into $H_j(u)$ by Step~(\ref{second_hash_step}).
    In this case, $u$ is marked as dormant in round $j$ by Step~(\ref{second_hash_step}).
    In both cases $j=0$ and $j > 0$, $u$ is marked as dormant in round $j\leq i-1$ which contradicts with the definition of $i$.
    So the only way for $u$ to become dormant in round $i$ is for a vertex $v$ to exist in $H_{i-1}(u)$ which is dormant in round $i-1$.
    Assume for contradiction that $r > 2^i$, then by Definition~\ref{def_r} there is no collision nor fully dormant vertex in $B(u, 2^i)$.
    By the induction hypothesis, there exists either a collision or a fully dormant vertex in $B(v,2^{i-1})$.
    By Lemma~\ref{radius_doubling}, we know that $H_{i-1}(u)=B(u, 2^{i-1}).$
    It means that $B(v,2^{i-1})\subseteq B(u,2^i)$ contains a collision or a fully dormant vertex, contradiction.
\end{proof}

%\begin{remark}\label{bbb}
    To state bounds simply,
    let $b \coloneqq \delta^{1/18}$, then hash functions $h_B$ and $h_V$ are from $[n]$ to $[m / b^{12}]$ and from $[n]$ to $[b^{6}]$, respectively.
    Note that $h_B$ and $h_V$ need to be independent with each other, but each being pairwise independent suffices, so each processor doing hashing only reads two words.
%\end{remark}

\begin{lemma}\label{ball_size_bound}
    For any vertex $u$ that is dormant after the \textsc{expand}, $|B(u, r)| \le b^2$ with probability at most $b^{-2}$.
\end{lemma}
\begin{proof}
    Let $j$ be the maximal integer such that $j \le d$ and $|B(u, j)| \le b^2$.
    We shall calculate the probability of $r \le j$, which is equivalent to the event $|B(u, r)| \le b^2$.

    The expectation of the number of collisions in $B(u, j)$ using $h_V$ is at most $\binom{b^2}{2} / b^{6} \le b^{-2} / 2$, then by Markov's inequality, the probability of at least one collision existing is at most $b^{-2} / 2$.

    For any $u$ to be fully dormant, at least one of the $n' - 1$ vertices other than $u$ must have hash value $h_B(u)$.
    By union bound, the probability of $u$ being fully dormant is at most
    \begin{equation}\label{ineq_prob_fully_dormant}
        \frac{n' - 1}{m / b^{12}} \le \frac{n'}{n' b^{6}} = \frac{1}{b^{6}},
    \end{equation}
    where the first inequality follows from $m / n' = b^{18}$.
    Taking union bound over all vertices in $B(u, j)$ we obtain the probability as at most $b^2 \cdot b^{-6} = b^{-4}$.

    Therefore, for any dormant vertex $u$, $\Pr[|B(u, r)| \le b^2] = \Pr[r \le j] \le b^{-2} / 2 + b^{-4} \le b^{-2}$ by a union bound.
\end{proof}

Based on Lemma~\ref{ball_size_bound}, we shall show that any dormant vertex has large enough neighbor set after the \textsc{expand} with good probability. To prove this, we need the following lemma:

\begin{lemma}\label{claim_ball_hash}
    $|B(u, r-1)| \le |H_i(u)|$.
\end{lemma}
\begin{proof}
    Let $w$ be a vertex in $B(u, r-1)$.
    If $\text{dist}(w, u) \le 2^{i-1}$ then $w \in H_{i-1}(u)$ by Lemma~\ref{radius_doubling} and Definition~\ref{def_r}, and position $h_V(w)$ in $H_i(u)$ remains occupied in round $i$.
    Suppose $\text{dist}(w, u) > 2^{i-1}$.
    Let $x$ be a vertex on the shortest path from $u$ to $w$ and $x$ satisfies $\text{dist}(x, w) = 2^{i-1}$.
    We obtain $B(x, 2^{i-1}) \subseteq B(u, r-1)$ since any $y \in B(x, 2^{i-1})$ has $\text{dist}(y, u) \le \text{dist}(y, x) + \text{dist}(x, u) \le 2^{i-1} + r-1 - 2^{i-1}$, and thus $w \in H_{i-1}(x)$ and $x \in H_{i-1}(u)$.
    So position $h_V(w)$ in $H_i(u)$ remains occupied in round $i$ as a consequence of Step~(\ref{second_hash_step}).
    Therefore Lemma~\ref{claim_ball_hash} holds.
\end{proof}

\begin{lemma}\label{hash_table_size_bound}
    After the \textsc{expand}, for any dormant vertex $u$, $|H(u)| < b$ with probability at most $b^{-1}$.
\end{lemma}
\begin{proof}
    By Lemma~\ref{ball_size_bound}, $|B(u, r)| > b^2$ with probability at least $1 - b^{-2}$.
    If this event happens, $u$ must be half dormant.
    In the following we shall prove that $|H(u)| \ge b$ with probability at least $1 - b^{-4}$ conditioned on this, then a union bound gives the lemma.

    If $i = 0$ then $r=1$ by Lemma~\ref{r_range}, which gives $|B(u, 1)| > b^2 \ge b$.
    Consider hashing arbitrary $b$ vertices of $B(u, 1)$.
    The expectation of the number of collisions among them is at most $b^2 / b^{6} = b^{-4}$, then by Markov's inequality the probability of at least one collision existing is at most $b^{-4}$.
    Thus with probability at least $1 - b^{-4}$ these $b$ vertices all get distinct hash values, which implies $|H(u)| \ge b$. So the lemma holds.

    Suppose $i \ge 1$.
    If $|B(u, r-1)| \ge b$ then the lemma holds by Lemma~\ref{claim_ball_hash}.
    Otherwise, since $|B(u, r)| > b^2$, by Pigeonhole principle there must exist a vertex $w$ at distance $r-1$ from $u$ such that $|B(w, 1)| > b$.
    Then by an argument similar to that in the second paragraph of this proof we have that $H_0(w) \ge b$ with probability at least $1 - b^{-4}$.
    If $i = 1$ then $r = 2$.
    Since $w \in H_0(u)$ dues to Definition~\ref{def_r}, at least $b$ positions will be occupied in $H_1(u)$, and the lemma holds.

    Suppose $i \ge 2$.
    Let $w_0$ be a vertex on the shortest path from $u$ to $w$ such that $\text{dist}(w_0, w) = 1$.
    Recursively, for $j \in [1, i-2]$, let $w_j$ be a vertex on this path such that $\text{dist}(w_{j-1}, w_j) = 2^j$.
    We have that
    $$\text{dist}(w, w_{i-2}) = \text{dist}(w, w_0) + \sum_{j \in [i-2]} \text{dist}(w_{j-1}, w_j) = 2^{i-1} - 1 .$$
    Thus
    $$\text{dist}(w_{i-2}, u) = r - 1 - \text{dist}(w, w_{i-2}) = r - 2^{i-1} \le 2^{i-1} ,$$
    where the last inequality follows from Lemma~\ref{r_range}.
    It follows that $w_{i-2} \in H_{i-1}(u)$ by Lemma~\ref{radius_doubling} and Definition~\ref{def_r}.
    We also need the following claim, which follows immediately from $\text{dist}(w_j, u) \le r - 2^j$ and Definition~\ref{def_r}:
    \begin{claim}\label{all_w_good}
        $H_j(w_j) = B(w_j, 2^j)$ for all $j \in [0, i-2]$.
    \end{claim}

    Now we claim that $|H_{i-1}(w_{i-2})| \ge b$ and prove it by induction, then $|H_i(u)| \ge b$ holds since $w_{i-2} \in H_{i-1}(u)$.
    $|H_1(w_0)| \ge b$ since $w \in H_0(w_0)$ and $|H_0(w)| \ge b$.
    Assume $|H_j(w_{j-1})| \ge b$.
    Since $w_{j-1} \in H_{j}(w_j)$ by $\text{dist}(w_{j-1}, w_j) = 2^j$ and Claim~\ref{all_w_good}, we obtain $|H_{j+1}(w_j)| \ge b$.
    So the induction holds and $|H_{i-1}(w_{i-2})| \ge b$.

    By the above paragraph and the first paragraph of this proof we proved Lemma~\ref{hash_table_size_bound}.
\end{proof}

\subsection{The Voting}\label{TRA}

In this section, we present the method \textsc{vote} and show that the number of ongoing vertices decreases by a factor of a positive constant power of $b$ with good probability.
%All vertices discussed are ongoing vertices.
\begin{framed}
\noindent \textsc{vote}: for each vertex $u$: initialize $u.l \coloneqq 1$,
\begin{enumerate}
    \item If $u$ is live after the \textsc{expand} then for each vertex $v$ in $H(u)$: if $v < u$ then set $u.l \coloneqq 0$. \label{case_live}
    \item Else set $u.l \coloneqq 0$ with probability $1 - b^{-2/3}$. \label{independent_sample}
\end{enumerate}
\end{framed}

There are two cases depending on whether $u$ is live.
In Case~(\ref{case_live}), by Lemma~\ref{radius_doubling}, $H(u)$ must contain all the vertices in the component of $u$, and so does any vertex in $H(u)$, because otherwise $u$ is dormant.
We need to choose the same parent for all the  vertices in this component, which is the minimal one in this component as described: a vertex $u$ that is not minimal in its component would have $u.l = 0$ by some vertex $v$ in $H(u)$ smaller than $u$. Thus all live vertices become finished in the next phase.

In Case~(\ref{independent_sample}), $u$ is dormant.
Then by Lemma~\ref{hash_table_size_bound}, $|H(u)| \ge b$ with probability at least $1 - b^{-1}$.
If this event happens,
%the expected number of leaders in $H(u)$ is at least $b \cdot b^{-2/3} = b^{1/3}$, thus by Chernoff bound
the probability of no leader in $H(u)$ is at most $(1 - b^{-2/3})^b \le \exp(- b^{-1/3}) \le b^{-1}$.
%$\exp(-b^{1/3} / 2) \le b^{-1}$.

%In the worst case of the number of  vertices in the next phase, the following happens in the current phase: (\romannumeral1) there is no live vertex; (\romannumeral2) all vertices fail to have $b$ vertices in their hash tables are set as leaders; (\romannumeral3)

%Recall that a \textsc{link} never updates the parent of a leader,
The number of vertices in the next phase is the sum of:
%(\romannumeral1) there is no live vertex;
(\romannumeral1) the number of dormant leaders,
(\romannumeral2) the number of non-leaders $u$ with $|H(u)| < b$ and no leader in $H(u)$,
and
(\romannumeral3) the number of non-leaders $u$ with $|H(u)| \ge b$ and no leader in $H(u)$.
%The number of  vertices in the next phase is at most the sum of the  number of leaders plus the number of vertices that satisfying (\romannumeral2) or (\romannumeral3).
We have that the expected number of  vertices in the next phase is at most
$$n' \cdot (b^{-2/3} + b^{-1} + (1 - b^{-1}) \cdot b^{-1}) \le n' \cdot b^{-1/2}.$$
By Markov's inequality, the probability of having more than $n' \cdot b^{-1/4}$ vertices in the next phase is at most
\begin{equation}\label{number_roots}
    \frac{n' \cdot b^{-1/2}}{n' \cdot b^{-1/4}} \le b^{-1/4}.
\end{equation}

\subsection{Removing the Assumption}\label{remove_assumption}

In this section, we remove Assumption~\ref{assumption1} which holds on a COMBINING CRCW PRAM, thus generalize the algorithm to run on an ARBITRARY CRCW PRAM.

Recall that we set $b = \delta^{1/18}$ where $\delta = m/n'$ is known.
The key observation is that in each phase, all results still hold when we use any $\tilde{n}$ to replace $n'$ as long as $\tilde{n} \ge n'$ and $b$ is large enough.
This effectively means that we use $b = (m / \tilde{n})^{1/18}$ for the hash functions $h_B$ and $h_V$.
This is because the only places that use $n'$ as the number of  vertices in a phase are:
\begin{enumerate}
    \item The probability of a dormant vertex $u$ having $B(u, r) \le b^2$ (Lemma~\ref{ball_size_bound}). Using $\tilde{n}$ to rewrite Inequality~(\ref{ineq_prob_fully_dormant}), we have:
        \begin{equation*}
            \frac{n' - 1}{m / b^{12}} = \frac{n' - 1}{\tilde{n} b^{6}} \le \frac{n'}{n' b^{6}} = \frac{1}{b^{6}},
        \end{equation*}
        followed from $m / \tilde{n} = b^{18}$ and $\tilde{n} \ge n'$.
        Therefore Lemma~\ref{ball_size_bound} still holds.
    \item The probability of having more than $n' \cdot b^{-1/4}$ vertices in the next phase.
        Now we measure the progress by the decreasing in $\tilde{n}$.
        The expected number of vertices in the next phase is still at most $n' \cdot b^{-1/4}$, where $n'$ is the exact number of vertices.
        Therefore by Markov's inequality we can rewrite Inequality~(\ref{number_roots}) as:
        $$\frac{n' \cdot b^{-1/2}}{\tilde{n} \cdot b^{-1/4}} \le b^{-1/4},$$
        which is the probability of having more than $\tilde{n} \cdot b^{-1/4}$ vertices in the next phase. \label{estimate_roots_bound}
\end{enumerate}

As a conclusion, if $\tilde{n} \ge n'$ and $b$ is large enough in each phase, all analyses still apply.
Let $c$ be the value defined in \textsc{prepare}.
We give the \emph{update rule} of $\tilde{n}$:
\begin{framed}
\noindent Update rule of $\tilde{n}$:

If $m / n \le \log^{c} n$ then set $\tilde{n} \coloneqq n / \log^{c} n$ for the first phase (after the \textsc{prepare}),
else set $\tilde{n} \coloneqq n$.

At the beginning of each phase, update $\tilde{n} \coloneqq \tilde{n} / b^{1/4}$ then update $b \coloneqq (m / \tilde{n})^{1/18}$.
\end{framed}

So $b \ge \log^{c/18} n$ is large enough.
By the above argument, we immediately have the following:
\begin{lemma}\label{nn_progress}
    Let $\tilde{n}$ and $n'$ be as defined above in each phase.
    If $\tilde{n} \ge n'$ in a phase, then with probability at least $1 - b^{-1/4}$, $\tilde{n} \ge n'$ in the next phase.
\end{lemma}

By an induction on phases, Lemma~\ref{nn_progress}, and a union bound, the following lemma is immediate:
\begin{lemma}\label{nn_prob}
    %Let $n'$ and $\tilde{n}$ be defined above.
    Given any integer $t \ge 2$,
    if $\tilde{n} \ge n'$ in the first phase, then
    $\tilde{n} \ge n'$ in all phases before phase $t$ with probability at least $1 - \sum_{i \in [t-2]} {b_{i}}^{-1/4}$, where $b_i$ is the parameter $b$ in phase $i \ge 1$ .
\end{lemma}

\subsection{Running Time}\label{rma}
In this section, we compute the running time of our algorithm and the probability of achieving it.

\begin{lemma}\label{running_time}
    After the \textsc{prepare}, if $\tilde{n} \ge n'$ in each phase, then the algorithm outputs the connected components in $O(\log\log_{m/n_1} n)$ phases,
    where $n_1$ is the $\tilde{n}$ in the first phase.
\end{lemma}
\begin{proof}
    Let $n_i$ be the $\tilde{n}$ in phase $i$.
    By the update rule of $\tilde{n}$,
    we have $n_{i+1} \le n_i / (m / n_i)^{1/72}$, which gives $m / n_{t+1} \ge (m / n_1)^{(73/72)^t}$.
    If $t = \lceil \log_{73/72} \log_{m/n_1} m \rceil + 1$ then $n_{t+1} < 1$, which leads to $n' = 0$ at the beginning of phase $t+1$.
    %The case $n' = 1$ is impossible since there is no edge for an ongoing vertex to exist.
    By Lemma~\ref{all_tree_flat} and monotonicity, the algorithm terminates and outputs the correct connected components in this phase since no parent changes.
\end{proof}

\begin{proof}[Proof of Theorem~\ref{main1}]
    We set $c = 100$ in the \textsc{prepare}.
    If $m / n > \log^{c} n$, then $\tilde{n} = n$ in the first phase by Lemma~\ref{prepare} and the update rule.
    Since $\delta = m / n_1 \ge \log^c n$, we have that $b \ge \delta^{1/18} \ge \log^5 n$ in all phases.
    By Lemma~\ref{nn_prob}, $\tilde{n} \ge n'$ in all phases before phase $\log n$ with probability at least $1 - \log n \cdot b^{-1/4} = 1 - 1 / \text{poly}(m \log n / n)$.
    If this event happens, by Lemma~\ref{running_time}, the number of phases is $O(\log\log_{m / n} n)$.
    By Lemma~\ref{expansion_rounds_bound}, the total running time is $O(\log d \log\log_{m / n} n)$ with good probability.

    If $m / n \le \log^{c} n$, then by Lemma~\ref{prepare} and the update rule of $\tilde{n}$, after the \textsc{prepare} which takes time $O(\log\log n)$, with probability at least $1 - 1/\log^c n$ we have $m / n_1 \ge \log^c n$.
    If this happens, by the argument in the previous paragraph, with probability at least $1 - 1 / \text{poly}(m / n_1 \cdot \log n)$, Connected Components algorithm takes time $O(\log d \log\log_{m / n} n)$.
    Taking a union bound, we obtain that with probability at least $1 - 1 / \text{poly}(m / n_1 \cdot \log n) - 1/\log^c n \ge 1 - 1 / \text{polylog}(n) = 1 - 1 / \text{poly}(m \log n / n)$, the total running time is $O(\log\log n) + O(\log d \cdot\log\log_{m / n_1} n) = O(\log d \log\log_{m / n} n)$.
\end{proof}

\section{Spanning Forest Algorithm}\label{sec_sf_alg}

Many existing connected components algorithm can be directly transformed into a spanning forest algorithm.
For example in Reif's algorithm \cite{reif1984optimal}, one can output the edges corresponding to leader contraction in each step to the spanning forest. However this is not the case here as we also add edges to the graph.
The solution in \cite{DBLP:conf/focs/Andoni} uses several complicated ideas including merging local shortest path trees, which heavily relies on computing the minimum function in constant time -- a goal easily achieved by sorting on an MPC but requiring $\Omega(\log\log n)$ time on a PRAM \cite{DBLP:conf/stoc/FichHRW85}.

We show how to modify our connected components algorithm with an extended expansion procedure to output a spanning forest.
Observe that in our previous expansion procedure, if a vertex $u$ does not stop expansion in step $i$, then the space corresponding to $u$ contains all the vertices within distance $2^i$ from $u$.
Based on this, we are able to maintain the distance from $u$ to the closest leader in $O(\log d)$ time by the distance doubling argument used before.
After determining the distance of each vertex to its closest leader,
for each edge $(u,v)$, if $u$ has distance $x$ to the closest leader, and $v$ has distance $x-1$ to the closest leader, then we can set $v$ as the parent of $u$ and add edge $(u,v)$ to the spanning forest.
(If there are multiple choices of $v$, we can choose arbitrary one.)
Since this parents assignment does not induce any cycle, we can find a subforest of the graph.
If we contract all vertices in each tree of such subforest to the unique leader in that tree (which also takes $O(\log d)$ time by shortcutting as any shortest path tree has height at most $d$), the problem reduces to finding a spanning forest of the contracted graph.
Similar to the analysis of our connected components algorithm, we need $O(\log d)$ time to find a subforest in the contracted graph and the number of contraction rounds is $O(\log\log_{m/n} n)$.
Thus, the total running time is asymptotically the same as our connected components algorithm.

Since the \textsc{expand} method will add new edges which are not in the input graph, our connected components algorithm cannot give a spanning forest algorithm directly. To output a spanning forest, we only allow direct links on graph arcs of the input graph. However, we want to link many graph arcs concurrently to make a sufficient progress. We extend the \textsc{expand} method to a new subroutine which can link many graph arcs concurrently. Furthermore, after applying the subroutine, there is no cycle induced by link operations, and the tree height is bounded by the diameter of the input graph.

For each arc $e$ in the current graph, we use $\hat{e}$ to denote the original arc in the input graph that is altered to $e$ during the execution.
Each edge processor is identified by an original arc $\hat{e}$ and stores the corresponding $e$ in its private memory during the execution.
To output the spanning forest, for a original graph arc $\hat{e} = (v, w)$, if at the end of the algorithm, $\hat{e}.f = 1$, then it denotes that the graph edge $\{v, w\}$ is in the spanning forest. Otherwise if both $\hat{e}.f=0$ and $\hat{e}'.f=0$ where $\hat{e}'$ denotes a graph arc $(w,v)$, then the edge $\{v,w\}$ is not in the spanning forest.

%In this section, unless otherwise specified, we still use the term ``tree'' to denote a tree represented by tree arcs $(v,v.p)$. For convenience, for a graph arc $e=(v,w),$ we use ``edge $e$'' to denote the undirected edge $\{v,w\}$.

\subsection{Vanilla Algorithm for Spanning Forest}

Firstly, let us see how to extend Vanilla algorithm to output a spanning forest.
The extended algorithm is called Vanilla-SF algorithm (see below).

The \textsc{random-vote} and \textsc{shortcut} are the same as those in Vanilla algorithm.
In the \textsc{alter} we only alter the current edge as in Vanilla algorithm but keep the original edge untouched.
We add a method \textsc{mark-edge} and an attribute $e$ for each vertex $v$ to store the current arc that causes the link on $v$, then $v.\hat{e}$ is the original arc in the input graph corresponding to $v.e$.
The \textsc{link} is the same except that we additionally mark the original arc in the forest using attribute $f$ if the corresponding current arc causes the link.
\begin{framed}
\noindent Vanilla-SF algorithm: repeat \{\textsc{random-vote}; \textsc{mark-edge}; \textsc{link}; \textsc{shortcut}; \textsc{alter}\} until no edge exists other than loops.

\textsc{mark-edge}: for each current graph arc $e=(v,w)$: if $v.l=0$ and $w.l=1$ then update $v.e$ to $e$ and update $v.\hat{e}$ to $\hat{e}$.

\textsc{link}: for each ongoing $u$: if $u.e=(u,w)$ exists then update $u.p$ to $w$ and update $u.\hat{e}.f \coloneqq 1$.
\end{framed}

The digraph defined by the labels is exactly the same as in Vanilla algorithm, therefore Lemma~\ref{all_tree_flat} holds for Vanilla-SF algorithm.
It is easy to see that Vanilla-SF algorithm uses $O(m)$ processors and can run on an ARBITRARY CRCW PRAM. Each phase takes $O(1)$ time.

%\begin{lemma}\label{all_tree_flat}
%At the beginning of each phase of forest algorithm Reif, each tree has height at most one.
%\end{lemma}
\begin{definition}\label{def_fj}
    For any positive integer $j$,
    at the beginning of phase $j$, let $F_j$ be the graph induced by all the edges corresponding to all the arcs $\hat{e}$ with $\hat{e}.f=1$.
\end{definition}
By the execution of the algorithm, for any positive integers $i \le j$, the set of the edges in $F_i$ is a subset of the set of the edges in $F_j$.
\begin{lemma}\label{lem:f_components}
    For any positive integer $j$,
    at the beginning of phase $j$, each vertex $u$ is in the component of $u.p$ in $F_j$.
\end{lemma}
\begin{proof}
    The proof is by induction. In the first phase, the lemma is vacuously true since $u.p = u$ for all vertices $u$. Now, suppose the lemma is true at the beginning of phase $i$. In phase $i$, if $u.p$ does not change, the claim is obviously true.
    Otherwise, there are two cases: (\romannumeral1) $u.p$ is changed in the \textsc{link}, or (\romannumeral2) $u.p$ is changed in the \textsc{shortcut}. If $u.p$ is changed in the \textsc{shortcut}, then by Lemma~\ref{all_tree_flat}, the original $u.p.p$ is changed in the \textsc{link}, and since $u$ and the original $u.p$ is in the same component of $F_{i}$, we only need to show that the \textsc{link} does not break the invariant.
    In the \textsc{link}, if $u.p$ is updated to $w$, there is a current graph arc $u.e = (u,w)$, and thus there exists a graph arc $u.\hat{e} = (x,y)$ in the input graph such that $x.p = u$ and $y.p = w$ at the beginning of phase $i$.
    By the induction hypothesis, $x$ and $u$ are in the same component of $F_{i}$ and thus of $F_{i+1}$. Similar argument holds for $y$ and $w$. Since $u.\hat{e}.f$ is set to $1$ in the \textsc{link}, $x$ and $y$ are in the same component in $F_{i+1}$. Thus, $u$ and $w$ are in the same component in $F_{i+1}$.
\end{proof}

\begin{lemma}\label{lem:induced_forest}
    For any positive integer $j$, $F_j$ is a forest.
    %At the beginning of each phase, all the edges $e$ with $e.f=1$ constitutes a forest.
    And in each tree of $F_j$, there is exactly one root.\footnote{The \emph{tree} we used here in the forest of the graph should not be confused with the tree in the digraph defined by labels.}
\end{lemma}
\begin{proof}
    By the \textsc{link}, every time the size of $\{u \mid u.p=u\}$ decreases by $1$, the size of $\{\hat{e} \mid \hat{e}.f = 1\}$ increases by $1$.
    Thus the size of $\{\hat{e} \mid \hat{e}.f = 1\}$ is exactly $n - |\{u \mid u.p=u\}|$, which induce at least $|\{u\mid u.p=u\}|$ components in $F_j$.
    By Lemma~\ref{lem:f_components}, there are at most $|\{u\mid u.p=u\}|$ components in $F_j$. %Since $|\{e\mid e.f=1\}|= n-|\{u\mid u.p=u\}|,$
    So there are exactly $|\{u\mid u.p=u\}|$ components in $F_j$ and each such component must be a tree, and each component contains exactly one vertex $u$ with $u.p=u$.
\end{proof}

\subsection{Algorithmic Framework}\label{subsec_tree_alg}

In this section, we show how to extend our Connected Components algorithm to a spanning forest algorithm.
\begin{framed}
\noindent Spanning Forest algorithm: \textsc{forest-prepare}; repeat \{\textsc{expand}; \textsc{vote}; \textsc{tree-link}; \textsc{tree-shortcut}; \textsc{alter}\} until no edge exists other than loops.

\textsc{forest-prepare}: if $m/n\leq \log^c n$ for given constant $c$ then run Vanilla-SF algorithm for $c\log_{8/7}\log n$ phases.

\textsc{tree-link}: for each ongoing $u$: update $u.p$, $u.e$, $u.\hat{e}$, and $u.\hat{e}.f$ according to some rule.

\textsc{tree-shortcut}: repeat \{\textsc{shortcut}\} until no parent changes.
\end{framed}

The \textsc{expand}, \textsc{vote}, \textsc{shortcut}, and \textsc{alter}  %\footnote{In our spanning forest algorithm, instead of only storing the final hash table $H(u)$ for each vertex $u$, the \textsc{expand} method stores hash table $H_i(u)$ for every round $i \ge 0$.}
are the same as in our connected component algorithm. %\textsc{expand} is almost as the same as in the connected component algorithm (see \S{\ref{exp_ana}}). The only difference is in step~\ref{first_hash_step}: instead of dealing with $v$ and $w$, we should deal with $v.p$ and $w.p$ for each directed edge $(v,w)$.

Similarly to that in Connected Components algorithm, the \textsc{forest-prepare} makes the number of ongoing vertices small enough with good probability.
As analyzed in \S{\ref{exp_ana}} and \S{\ref{TRA}}, the \textsc{expand} takes $O(\log d)$ time, and \textsc{vote} takes $O(1)$ time. \textsc{tree-shortcut} takes $O(\log d')$ time where $d'$ is the height of the highest tree after the \textsc{tree-link}. Later, we will show that the height $d'$ is $O(d)$, giving $O(\log d)$ running time for each phase.

\subsection{The Tree Linking}
In this section, we present the method \textsc{tree-link}.
The purpose of the \textsc{tree-link} is two-fold:
firstly, we want to add some edges to expand the current forest (a subgraph of the input graph);
secondly, we need the information of these added edges to do link operation of ongoing vertices to reduce the total number of ongoing vertices.
For simplicity, all vertices discussed in this section are ongoing vertices in the current phase.

Similar to the \textsc{expand}, we shall use a pool of $m$ processors to do the \textsc{tree-link}. Let $n'$ be the number of vertices. We set all the parameters as the same as in the \textsc{expand}: $\delta = m/n'$, the processors are divided into $m/\delta^{2/3}$ indexed blocks where each block contains $\delta^{2/3}$ indexed processors, and both $h_B$, $h_V$ are the same hash functions used in the \textsc{expand}. %Similarly, we use $H(u)$ to denote the hash table of the vertex $u$.
Comparing to Connected Components algorithm, we store not only the final hash table $H(u)$ of $u$, but also the hash table $H_j(u)$ of $u$ in each round $j \ge 0$. (In Connected Components algorithm, $H_j$ is an analysis tool only.)
Let $T$ denote the total number of rounds in Step~(\ref{expand_loop}) of the \textsc{expand}.

We present the method \textsc{tree-link} as follows, which maintains:
(\romannumeral1) the largest integer $u.{\alpha}$ for each vertex $u$ such that there is neither collisions, leaders, nor fully dormant vertices in $B(u, u.{\alpha})$;
(\romannumeral2) $u.{\beta}$ as the distance to the nearest leader $v$ (if exists in $B(u, u.{\alpha} + 1)$) from $u$;
(\romannumeral3) a hash table $Q(u)$ to store all vertices in $B(u, u.{\alpha})$, which is done by reducing the radius by a factor of two in each iteration and attempting to expand the current $Q(u)$ to a temporary hash table $Q'(u)$.

\begin{framed}
\noindent \textsc{tree-link}:
\begin{enumerate}
%\item Each vertex is either \emph{live} or \emph{dormant} in a step. Mark every vertex as \emph{live} at the beginning.
%\item Map the roots to blocks using $h_B$. Mark the roots that do not own a block as \emph{dormant}.\label{it:full_dormant_tree}
%\item For each directed edge $(v,w):$ if $v$ is live before Step~\ref{it:full_dormant_tree}, then use $h_V$ to hash $v$ into $H_0(v)$ and $w$ into $H_0(v),$ else mark $w$ as dormant.
%\item For each hash done in Step~\ref{it:full_dormant_tree}: if it causes a \emph{collision} in $H_0(u)$ then mark $u$ as \emph{dormant}.
%\item Repeat until there is neither live vertex nor hash table getting a new entry: \label{it:doubling_ball}
%\begin{enumerate}
% \item[] round $i$:
% \item For each $u$: for each $v$ in $H_{i-1}(u)$: if $v$ is dormant before round $i$, then mark $u$ as dormant; for each $w$ in $H_{i-1}(v)$: use $h_V$ to hash $w$ into $H_i(u).$
% \item For each hash done in round $i$: if it causes a collision in $H_i(u)$, then mark $u$ as dormant.
%\end{enumerate}

\item For each vertex $u$: \label{it:init_tree}
\begin{enumerate}
\item If $u.l=1$, set $u.{\alpha} \coloneqq -1$ and set hash table $Q(u) \coloneqq \emptyset$.
\item If $u.l=0$ and $u$ does not own a block, set $u.{\alpha} \coloneqq -1$ and set $Q(u) \coloneqq \emptyset$.
\item If $u.l=0$ and $u$ owns a block, set $u.{\alpha} \coloneqq 0$ and use $h_V$ to hash $u$ into $Q(u)$.
\end{enumerate}

\item For $j = T\rightarrow 0$: for each vertex $u$ with $u.{\alpha}\geq 0$: if every $v$ in table $Q(u)$ is live in round $j$ of Step~(\ref{expand_loop}) of the \textsc{expand}: \label{it:expand_j}
\begin{enumerate}
\item Initialize $Q'(u) \coloneqq \emptyset$.
\item For each $v$ in $Q(u)$: for each $w$ in $H_j(v)$: use $h_V$ to hash $w$ into $Q'(u)$.
\item If there is neither collisions nor leaders in $Q'(u)$ then set $Q(u)$ to be $Q'(u)$ and increase $u.{\alpha}$ by $2^j$. \label{it:set_Qu}
\end{enumerate}

\item For each current graph arc $(v,w)$: if $v.l=1$ then mark $w$ as a \emph{leader-neighbor}.
%\item For each vertex $u$: initialize $u.{\beta} = \infty$. ($u.{\beta}$ maintains the distance from $u$ to one of the closest leader.)
\item For each vertex $u$:
\begin{enumerate}
\item If $u.l=1$ then set $u.{\beta} \coloneqq 0$.
\item If $u.l=0$ and $Q(u)$ contains a vertex $w$ marked as a leader-neighbor then set $u.{\beta} \coloneqq u.{\alpha} + 1$. \label{it:set_ub}
\end{enumerate}
\item For each current graph arc $e=(v,w)$: if $v.{\beta}=w.{\beta}+1$ then write $e$ into $v.e$ and write $\hat{e}$ to $v.\hat{e}$.
    \label{step_update_b}
\item For each vertex $u$: if $u.e=(u,w)$ exists then update $u.p$ to $w$ and update $u.\hat{e}.f \coloneqq 1$. \label{step_update_p}

\end{enumerate}
\end{framed}

For simplicity, if there is no ambiguity, we also use $Q(u)$ to denote the set of vertices which are stored in the table $Q(u)$.
We call each iteration $j$ from $T$ down to $0$ in Step~(\ref{it:expand_j}) a \emph{round}.

\begin{lemma}\label{lem:property_of_Qu}
    In any round in the \textsc{tree-link}, for any vertex $u$, $Q(u) = B(u,u.{\alpha})$. Furthermore, at the end of the \textsc{tree-link}, $u.{\alpha}$ is the largest integer such that there is neither collisions nor dormant vertices in $B(u, u.{\alpha})$.

%$\forall v,w\in B(u,u.{\alpha}), h_V(v)\not = h_V(w)$ and $\forall v\in B(u,u.{\alpha}),$ it satisfies $v.l=0$, and $v$ owns a block.
\end{lemma}
\begin{proof}
Firstly, we show that $Q(u)=B(u,u.{\alpha})$. Our proof is by an induction on $j$ of Step~(\ref{it:expand_j}). The base case holds since before Step~(\ref{it:expand_j}), the initialization of $Q(u)$ and $u.{\alpha}$ satisfy the claim.
Now suppose $Q(u)=B(u,u.{\alpha})$ at the beginning of round $j$ of Step~(\ref{it:expand_j}). There are two cases.
The first case is that $Q(u)$ does not change in this round. In this case the claim is true by the induction hypothesis.
The second case is that $Q(u)$ will be set to $Q'(u)$ in Step~(\ref{it:set_Qu}). Since there is no collision in $Q'(u)$, we have $Q'(u)=\bigcup_{v\in Q(u)} H_{j}(v)$, where $Q(u)$ is not set in Step~(\ref{it:set_Qu}) yet.
Since every $v$ in $Q(u)$ is live in round $j$ of Step~(\ref{expand_loop}) of the \textsc{expand}, we have $H_{j}(v)=B(v,2^{j})$ according to Lemma~\ref{radius_doubling}.
Together with the induction hypothesis, $Q'(u)=\bigcup_{v\in B(u,u.{\alpha})} B(v,2^{j})=B(u,u.{\alpha}+2^{j})$. Thus by Step~(\ref{it:set_Qu}), after updating $Q(u)$ and $u.{\alpha}$, the first part of the lemma holds.
Now we prove the second part of the lemma.
For convenience in the notation, we say that $B(u, \alpha)$ satisfies property $\mathcal{P}$ if and only if
there is neither collisions, leaders, nor fully dormant vertices in $B(u, \alpha)$.
%(\romannumeral1) $\forall v,w\in B(u,u.{\alpha}),$ $h_V(v)\not = h_V(w)$;
%(\romannumeral2) $\forall v\in B(u,u.{\alpha}),$ it satisfies $v.l=0$;
%(\romannumeral3) $\forall v\in B(u,u.{\alpha}),$ $v$ owns a block.
We shall show that at the end of the \textsc{tree-link}, $u.{\alpha}$ is the largest integer such that $B(u, u.{\alpha})$ satisfies $\mathcal{P}$.
For any vertex $u$, if $u.l = 0$ and $u$ is fully dormant, or $u.l=1$, then the claim holds due to Step~(\ref{it:init_tree}).
Consider a vertex $u$ with $u.l=0$ and $u$ owning a block.
Our proof is by induction on $j$ of Step~(\ref{it:expand_j}).
We claim that at the end of round $j$ of Step~(\ref{it:expand_j}), the following two invariants hold:
(\romannumeral1) $B(u, u.{\alpha}+2^j)$ does not satisfy $\mathcal{P}$; (\romannumeral2) $B(u, u.{\alpha})$ satisfies $\mathcal{P}$.
The base case is before Step~(\ref{it:expand_j}).
There are two cases:
if $u$ is live after round $T$ of Step~(\ref{expand_loop}) of the \textsc{expand}, then there is a leader in $B(u,2^{T+1})$ by the \textsc{vote};
otherwise, there is either a fully dormant vertex or a collision in $B(u,2^{T+1})$.
Thus, the invariants hold for the base case.

Now suppose the invariants hold after round $j$ of Step~(\ref{it:expand_j}). In round $j-1$, there are two cases.
In the first case, $Q(u)$ is set to $Q'(u)$.
This means that before Step~(\ref{it:set_Qu}), $\forall v\in Q(u)=B(u,u.{\alpha})$, $B(v,2^{j-1})$ satisfies $\mathcal{P}$. Thus, $B(u,u.{\alpha}+2^{j-1})$ satisfies $\mathcal{P}$. Notice that by the induction, $B(u,u.{\alpha}+2^{j})$ does not satisfy $\mathcal{P}$. Therefore, the invariants hold after updating $Q(u)$ and $u.{\alpha}$ in Step~(\ref{it:set_Qu}).
In the second case, $Q(u)$ remains unchanged. There exists $v\in Q(u) = B(u,u.{\alpha})$ such that $B(v, 2^{j-1})$ does not satisfy $\mathcal{P}$ which means that $B(u,u.{\alpha}+2^{j-1})$ does not satisfy $\mathcal{P}$.
Since $Q(u)$ and $u.{\alpha}$ can only change together, $B(u,u.{\alpha})$ satisfies $\mathcal{P}$.
The invariants also hold.
Therefore, after round $j=0$, $B(u,u.{\alpha})$ satisfies $\mathcal{P}$ and $B(u,u.{\alpha}+1)$ does not satisfy $\mathcal{P}$, giving the second part of the lemma.
\end{proof}

\begin{lemma}\label{lem:u.b}
For any vertex $u$, if $u.{\beta}$ is updated, then $u.{\beta}=\underset{v: v.l=1}{\min}\text{dist}(u,v)$.
\end{lemma}
\begin{proof}
For any vertex $u$ with $u.l=1$, $u.{\beta}$ is set to $0$.
Now, consider a vertex $u$ with $u.l=0$.
By Lemma~\ref{lem:property_of_Qu}, $Q(u) = B(u,u.{\alpha})$ and there is no leader in $B(u,u.{\alpha})$.
If $w \in B(u, u.{\alpha})$ is marked as a leader-neighbor, there is a vertex $v$ with $v.l = 1$ such that $v \in B(u, u.{\alpha}+1)$.
Thus, in this case $u.{\beta}=u.{\alpha}+1=\min_{v:v.l=1}\text{dist}(u,v)$.
\end{proof}

The following lemma shows that the construction of the tree in Steps~(\ref{step_update_b},\ref{step_update_p}) using the $\beta$ values is valid.

\begin{lemma}
    For any vertex $u$, if $u.{\beta} > 0$, then there exists an edge $\{u, w\}$ such that $w.{\beta} = u.{\beta} - 1$.
\end{lemma}
\begin{proof}
Consider a vertex $u$ with $u.{\beta}=1$. Then $u.{\alpha} = 0$, and by Lemma~\ref{lem:property_of_Qu}, $Q(u)=B(u,0)=\{u\}$. Thus, $u$ is marked as a leader-neighbor which means that there is a graph edge $\{u,v\}$ where $v$ is a leader. Notice that $v.{\beta}=0$. So the lemma holds for $u$ with $u.{\beta}=1$.

Consider a vertex $u$ with $u.{\beta} > 1$. By Lemma~\ref{lem:u.b}, there is a leader $v$ such that $\text{dist}(u, v) = u.{\beta}$. Let $\{u,w\}$ be a graph edge with $\text{dist}(w, v) = u.{\beta} - 1$, which must exist by Step~(\ref{step_update_b}).
It suffices to show that $w.{\beta} = u.{\beta}-1$.
Since $B(w, u.{\beta}-1)$ contains a leader $v$, and $B(w, u.{\beta}-2)\subseteq B(u, u.{\beta}-1) = B(u,u.{\alpha})$ which does not contain a leader by Lemma~\ref{lem:property_of_Qu}, we have that $w.{\alpha} = u.{\beta}-2$.
Let $\{x, v\}$ be a graph edge such that $\text{dist}(w, x)=u.{\beta}-2$, which must exist by Step~(\ref{step_update_b}).
Then, $x$ is marked as a leader-neighbor and $x$ is in $B(w, w.{\alpha})$.
Hence $w.{\beta} = w.{\alpha}+1 = u.{\beta}-2 + 1 = u.{\beta} -1$.
\end{proof}

The above lemma implies that if $u.{\beta}>0$ then $u$ is a non-root in the next phase due to Steps~(\ref{step_update_b},\ref{step_update_p}):
\begin{corollary}\label{cor:u.bge0}
    For any vertex $u$, if $u.{\beta}>0$, then $u$ is finished in the next phase.
\end{corollary}

\begin{lemma}
    The height of any tree is $O(d)$ after the \textsc{tree-link}.
\end{lemma}
\begin{proof}
    If $u.p$ is updated to $w$, then $u.{\beta}=w.{\beta} + 1$ by Step~(\ref{step_update_p}).
    Thus, the height of a tree is at most $\max_u u.{\beta} + 1$. By Lemma~\ref{lem:u.b} and the fact that an \textsc{alter} and adding edges never increase the diameter, the height of any tree is at most $d$.
\end{proof}
As mentioned at the end of \S{\ref{subsec_tree_alg}}, by the above lemma and Lemma~\ref{expansion_rounds_bound}, the following is immediate:
\begin{corollary}\label{phase_time}
    Each phase of the algorithm takes $O(\log d)$ time.
\end{corollary}

Similar to Definition~\ref{def_fj} and Lemma~\ref{lem:induced_forest}, we can show the following, which guarantees the correctness of Spanning Forest algorithm:
\begin{lemma}
    At the end of the \textsc{tree-link}, all the edges $\hat{e}$ with $\hat{e}.f=1$ constitute a forest. And in each tree of the forest, there is exactly one root.
\end{lemma}

\subsection{Running Time}

Now let us analyze the number of vertices in the next phase. If a vertex $u$ is live after the \textsc{expand}, then it is finished in the next phase. Thus all the vertices in the next phase are dormant after the \textsc{expand} in this phase.
Consider a dormant vertex $u$,
and let $r$ be that defined in Definition~\ref{def_r} with respect to $u$.

\begin{lemma}
    For a dormant non-leader $u$, if there is a leader in $B(u,r)$, then $u$ is finished in the next phase.
\end{lemma}
\begin{proof}
    Let $v$ be the leader closest to $u$. Since $v$ is in $B(u,r)$, by the definition of $r$ and Lemma~\ref{lem:property_of_Qu}, we have that $u.{\alpha} = \text{dist}(u,v)-1$, and $Q(u)=B(u,\text{dist}(u,v)-1)$ which contains a leader-neighbor.
    By Step~(\ref{it:set_ub}) in the \textsc{tree-link}, $u.{\beta}=\text{dist}(u,v)>0$. By Corollary~\ref{cor:u.bge0}, $u$ is finished in the next phase.
\end{proof}

Using the above lemma and Corollary~\ref{phase_time},
the remaining analysis is almost identical to that in Connected Components algorithm.

\begin{proof}[Proof of Theorem~\ref{main2}]
    %See \S{\ref{TRA}}, \S{\ref{remove_assumption}}, and \S{\ref{rma}}.

By Lemma~\ref{ball_size_bound}, we have that $|B(u,r)|\leq b^2$ with probability at most $b^{-2}$.
Conditioned on $|B(u,r)|> b^2$, the probability that $B(u,r)$ contains no leader is at most $b^{-1}$.
The number of vertices in the next phase is at most the sum of
%1) the number of dormant leader vertices, 2) the number of vertices $u$ with $|B(u,r)|\leq b^2$, and 3) the number of vertices $u$ with $|B(u,r)|>b^2$ and no leader in $B(u,r)$.
(\romannumeral1) the number of dormant leaders,
(\romannumeral2) the number of vertices $u$ with $|B(u,r)|\leq b^2$,
and
(\romannumeral3) the number of vertices $u$ with $|B(u,r)|>b^2$ and no leader in $B(u,r)$.
Thus, the probability that a dormant vertex $u$ is in the next phase is at most $b^{-2/3} + b^{-2} +(1-b^{-2})\cdot b^{-1}\leq b^{-1/2}$.
By Markov's inequality, the probability of having more than $n'\cdot b^{-1/4}$ vertices in the next phase is at most $b^{-1/4}$.

Finally, using exactly the same analyses in \S{\ref{remove_assumption}} and \S{\ref{rma}}, we obtain that
the algorithm runs on an ARBITRARY CRCW PRAM and outputs the spanning forest.
Moreover, with probability $1-1/\text{poly}(m \log n / n)$, the number of phases is $O(\log\log_{m / n} n)$ thus the total running time is $O(\log d \log\log_{m / n} n)$ (cf. Corollary~\ref{phase_time}).
\end{proof}

\section{Faster Connected Components Algorithm: Full Proof of An $O(\log d + \log\log_{m/n} n)$ Bound} \label{sec_faster_cc_alg}

In this section we give the full proof of Theorem~\ref{main3} by presenting a faster algorithm (see below) for connected components with a detailed analysis.

\begin{framed}
\noindent Faster Connected Components algorithm: \textsc{compact}; repeat \{\textsc{expand-maxlink}\} until the graph has diameter $O(1)$ and all trees are flat;
run Connected Components algorithm from \S{\ref{sec_cc_alg}}.
%\textsc{tree-shortcut}.

%\textsc{compact}: \textsc{prepare}; %If $m / n \le \log^c n$ for given constant $c$ then run $c \log_{8/7}\log n$ phases of Vanilla algorithm; rename each ongoing vertex to a distinct id in $[2 m / \log^c n]$ and assign it a block of size $\max\{m / n, \log^c n\} / \log^2 n$.
\end{framed}

Each iteration of the repeat loop is called a \emph{round}.
The \emph{break condition} that the graph has diameter $O(1)$ and all trees are flat shall be tested at the end of each round.
The main part of this section is devoted to the repeating \textsc{expand-maxlink}.
Before that, we explain the method \textsc{compact}.

\begin{framed}
\noindent \textsc{compact}: \textsc{prepare}; %If $m / n \le \log^c n$ for given constant $c$ then run $c \log_{8/7}\log n$ phases of Vanilla algorithm;
rename each ongoing vertex to a distinct id in $[2 m / \log^c n]$ and assign it a block of size $\max\{m / n, \log^c n\} / \log^2 n$.
\end{framed}

%\paragraph{Explanation of \textsc{compact}.}
Recall that in the \textsc{prepare} in \S{\ref{sec_cc_alg_framwork}}, it is guaranteed that the number of ongoing vertices is at most $m / \log^c n$ with good probability, then we use hashing to assign them blocks. The problem of this method is that in this section, we need, with high probability, all roots to own blocks as we need to union bound (as the events are dependent) all roots on all shortest paths (see \S{\ref{sec_overview}}), but the hashing method only gives that each root owns a block with probability $1 - 1 / \text{polylog}(n)$.
To solve this problem, we need the following tool when implementing \textsc{compact}:

\begin{definition}\label{def_apx_compaction}
    Given a length-$n$ array $A$ which contains $k$ \emph{distinguished} elements, \emph{approximate compaction} is to map all the distinguished elements in $A$ one-to-one to an array of length $2k$.
\end{definition}

\begin{lemma}[\cite{DBLP:conf/focs/Goodrich91}] \label{lem_apx_compaction}
    There is an ARBITRARY CRCW PRAM algorithm for approximate compaction that runs in $O(\log^* n)$ time and uses $O(n)$ processors with high probability.
    Moreover, if using $n \log n$ processors, the algorithm runs in $O(1)$ time.\footnote{\cite{DBLP:conf/focs/Goodrich91} actually compacts the $k$ distinguished elements of $A$ to an array of length $(1 + \epsilon) k$ using $O(n / \log^* n)$ processors with probability $1 - 1 / c^{n^{1/25}}$ for any constants $\epsilon > 0$ and $c > 1$, but Lemma~\ref{lem_apx_compaction} suffices for our use. The second part of the lemma is a straightforward corollary (cf. \S{2.4.1} in \cite{DBLP:conf/focs/Goodrich91}).}
\end{lemma}

\begin{lemma}\label{lem_stage_1}
    With good probability, \textsc{compact} renames each ongoing vertex a distinct id in $[2 m / \log^c n]$, assigns each of them a block of size $\max\{m / n, \log^c n\} / \log^2 n$, runs in $O(\log\log_{m / n} n)$ time, and uses $O(m)$ processors in total.
\end{lemma}
\begin{proof}
    If $m / n > \log^c n$, then the \textsc{compact} does not run Vanilla algorithm and each (ongoing) vertex already has a distinct id in $[n] \subseteq [2m / \log^c n]$. Since there are $\Theta(m)$ processors in total, it is easy to assign each vertex a block of size $m / n / \log^2 n$, which is $\max\{m / n, \log^c n\} / \log^2 n$.

    Else if $m / n \le \log^c n$, the \textsc{compact} runs Vanilla algorithm for $c \log_{8/7}\log n = O(\log\log_{m / n} n)$ phases such that the number of ongoing vertices is at most $m / \log^c n$ with good probability (cf. Corollary~\ref{van_prob}).
    Conditioned on this happening, the \textsc{compact} uses approximate compaction to compact all the at most $m / \log^c n$ ongoing vertices (distinguished elements) to an array of length at most $2m / \log^c n$, whose index in the array serves as a distinct id.
    By Lemma~\ref{lem_apx_compaction}, this runs in $O(\log^* n)$ time with probability $1-1/\text{poly}(n)$.
    Once the ongoing vertices are indexed, it is easy to assign each of them a block of size $m / (m / \log^c n) / \log^2 n$, which is $\max\{m / n, \log^c n\} / \log^2 n$.
    By a union bound, the \textsc{compact} takes $O(\log\log_{m / n} n)$ time with good probability.
\end{proof}

The renamed vertex id is used only for approximate compaction; for all other cases, we still use the original vertex id.
Renaming the vertex id to range $[2 m / \log^c n]$ not only facilitates the processor allocation at the beginning, but more importantly, guarantees that each subsequent processor allocation takes $O(1)$ time, because the array $A$ to be compacted has length $|A| \le 2 m / \log^c n$ while the number of processors is at least $|A| \log |A|$ when $c \ge 10$ (cf. Lemma~\ref{lem_apx_compaction}, see \S{\ref{sec_overview}} and details in \S{\ref{sec_faster_alg_framework}}).

In the remainder of this section,
we present the detailed method \textsc{expand-maxlink} in \S{\ref{sec_faster_alg_framework}},
then we prove that Faster Connected Components algorithm correctly computes the connected components of the input graph (cf. \S{\ref{sec_correctness}}), each round (\textsc{expand-maxlink}) can be implemented on an ARBITRARY CRCW PRAM in constant time (cf. \S{\ref{sec_imp}}),
uses $O(m)$ processors over all rounds (cf. \S{\ref{sec_space}}),
reduces the diameter to $O(1)$ and flatten all trees in $O(\log d + \log\log_{m/n} n)$ rounds (cf. \S{\ref{sec_diameter_reduction}}),
leading to the proof of Theorem~\ref{main3} in \S{\ref{sec_pf_main3}}.

\subsection{Algorithmic Framework} \label{sec_faster_alg_framework}

%Stage~(\ref{stage_2}) runs $O(\log d + \log\log_{m / n} n)$ \emph{rounds} of the \textsc{expand-maxlink} (see below for the algorithmic framework of a round).
In this section, we present the key concepts and algorithmic framework of the \textsc{expand-maxlink}.

\paragraph{Level and budget.}
The \emph{level} $\ell(v)$ of a vertex $v$ is a non-negative integer that can either remain the same or increase by one during a round.
By Lemma~\ref{lem_stage_1}, at the beginning of round $1$, each ongoing vertex $v$ owns a block of size $b_1 \coloneqq \max\{m / n, \log^c n\} / \log^2 n$, whose level is defined as $1$;
the level of a non-root vertex whose parent is ongoing is defined as $0$ and $b_0 \coloneqq 0$ for soundness;
all other vertices are ignored as their components have been computed.
%(The vertices with level $0$ have no incident edges and are non-roots, thus never participate in the execution of Stage~(\ref{stage_2}); their levels are defined for soundness.)
During a round, some roots become non-roots by updating their parents; for those vertex remains to be a root, its level might be increased by one;
a root with level $\ell$ is assigned to a block of size $b_{\ell} \coloneqq b_1^{1.01^{\ell - 1}}$ at the end of the round.
Given $b \ge 0$, a vertex $v$ has \emph{budget} $b(v) \coloneqq b$ if the maximal-size block owned by $v$ has size $b$. %(so the budget of any vertex cannot decrease).
Each block of size $b$ is partitioned into $\sqrt{b}$ (indexed) tables, each with size $\sqrt{b}$.

\paragraph{Neighbor set.}
Recall from \S{\ref{sec_cc_alg_framwork}} that the edges that define the current graph include:
(\romannumeral1) the (altered) original edges corresponding to edge processors, and (\romannumeral2) the (altered) added edges in the tables over all rounds of all vertices;
any vertex within distance $1$ from $v$ (including $v$) in the current graph is called a neighbor of $v$.
For any vertex $v$, let $N(v)$ be the set of its neighbors.
%Note that given $v$ and $N(v)$ at the beginning of a round, $N(v)$ does not change before Step~(\ref{step_3}), might obtain new elements during Steps (\ref{step_3}-\ref{step_5}), and never lose elements before the \textsc{alter} in Step~(\ref{step_6}).
In Step~(\ref{step_3}) we use the old $N(v)$ when initializing the loop that enumerates $N(v)$.
For any vertex set $S$, define $N(S) \coloneqq \bigcup_{w \in S} N(w)$,
and define $S.p \coloneqq \{w.p \mid w \in S \}$.

\paragraph{Hashing.}
At the beginning of a round,
one random hash function $h$ is chosen, then all neighbor roots of all roots use $h$ to do individual hashing in Step~(\ref{step_3}).
As same as in \S{\ref{sec_cc_alg}}, a pairwise independent $h$ suffices, thus each processor only reads two words.
The hashing in Step~(\ref{step_5}) follows the same manner using the same $h$.
%If Step~(\ref{step_8}) is successful, which we will show with good probability, each vertex has a new block for each round.
For each vertex $v$, let $H(v)$ be the first table in its block, which will store the added edges incident on $v$.
Step~(\ref{step_5}) is implemented by storing the old tables for all vertices while hashing new items (copied from the old $H(v)$ and old $H(w)$ in the block of $w$) into the new table.
%In Step~(\ref{step_8}), a root $v$ might get assigned a new (empty) block, but the name $H(v)$ in our analysis always refers to the first table in the block before Step~(\ref{step_8}).

\begin{framed}
\noindent \textsc{expand-maxlink}:
\begin{enumerate}
    \item \textsc{maxlink}; \textsc{alter}. \label{step_1}
    \item For each root $v$: increase $\ell(v)$ with probability $b(v)^{-0.06}$. \label{step_2}
    \item For each root $v$: for each root $w \in N(v)$: if $b(w) = b(v)$ then hash $w$ into $H(v)$. \label{step_3}
    \item For each root $v$: if there is a collision in $H(v)$ then mark $v$ as dormant.
        For each vertex $v$: if there is a dormant vertex in $H(v)$ then mark $v$ as dormant. \label{step_4}
    \item For each root $v$: for each $w \in H(v)$: for each $u \in H(w)$: hash $u$ into $H(v)$. For each root $v$: if there is a collision in $H(v)$ then mark $v$ as dormant. \label{step_5}
    \item \textsc{maxlink}; \textsc{shortcut}; \textsc{alter}. \label{step_6}
    \item For each root $v$: if $v$ is dormant and did not increase level in Step~(\ref{step_2}) then increase $\ell(v)$. \label{step_7}
    \item For each root $v$: assign a block of size $b_{\ell(v)}$ to $v$. \label{step_8}
    %\item \textsc{shortcut}. \label{step_9}
\end{enumerate}

\noindent \textsc{maxlink}: repeat \{for each vertex $v$: let $u \coloneqq \argmax_{w \in N(v).p} \ell(w)$, if $\ell(u) > \ell(v)$ then update $v.p$ to $u$\} for $2$ iterations.

\end{framed}

The \textsc{alter} and \textsc{shortcut} (cf. Steps (\ref{step_1},\ref{step_6})) are the same as in \S{\ref{sec_cc_alg}}.
The \textsc{maxlink} uses parent links, instead of direct links in \S{\ref{sec_cc_alg}} and \S{\ref{sec_sf_alg}}.

%A (weak) bound on the number of rounds and maximum level is presented below, which shall be used in many places. (Later in \S{\ref{sec_correctness}} we will prove a stronger bound on the maximum level.)
%\begin{lemma} \label{lem_round_level}
%    Stage~(\ref{stage_2}) runs for at most $O(\log n)$ rounds and the level of any vertex is at most $O(\log n)$.
%\end{lemma}
%\begin{proof}
%    \textsc{expand-maxlink} is run for $O(\log d + \log\log_{m / n} n) \le O(\log n)$ rounds and each vertex can increase its level by at most $1$ in any round (cf. Steps (\ref{step_2},\ref{step_7})).
%\end{proof}

\subsection{Correctness} \label{sec_correctness}

In this section, we prove that Faster Connected Components algorithm correctly computes the connected components of the input graph when it ends.

First of all, we prove a useful property on levels and roots.

\begin{lemma} \label{lem_non_root}
    If a vertex $v$ is a non-root at any step, then during the execution after that step,
    $v$ is a non-root, $\ell(v)$ cannot change, and $0 \le \ell(v) < \ell(v.p)$.
\end{lemma}
\begin{proof}
    The proof is by an induction on rounds.
    At the beginning of the first round, by the definition of levels (cf. \S{\ref{sec_faster_alg_framework}}), each ongoing vertex $v$ has level $1$ and $v = v.p$, and each non-root with an ongoing parent has level $0$ (vertices in trees rooted at finished vertices are ignored), so the lemma holds.

    In Step~(\ref{step_1}), each iteration of a \textsc{maxlink} can only update the parent to a vertex with higher level, which cannot be itself.
    In Step~(\ref{step_2}), level increase only applies to roots.
    The invariant holds after the \textsc{maxlink} in Step~(\ref{step_6}) for the same reasons as above.
    In \textsc{shortcut} (cf. Step~(\ref{step_6})), each vertex $v$ updates its parent to $v.p.p$, which, by the induction hypothesis, must be a vertex with level higher than $v$ if $v$ is a non-root, thus cannot be $v$.
    In Step~(\ref{step_7}), level increase only applies to roots.
    All other steps and \textsc{alter}s do not change the labeled digraph nor levels, giving the lemma.
\end{proof}

Based on this, we can prove the following key result, which implies the correctness of our algorithm.
\begin{lemma} \label{lem_tree_invariant}
    The following conditions hold at the end of each round:
    \begin{enumerate}
        \item For any component in the input graph, its vertices are partitioned into trees in the labeled digraph such that each tree belongs to exactly one component. \label{condition_1}
        \item A tree does not contain all the vertices in its component if and only if there is an edge between a vertex in this tree and a vertex in another tree. \label{condition_2}
    \end{enumerate}
\end{lemma}

The proof of Lemma~\ref{lem_tree_invariant} relies on the following invariant:
\begin{lemma}\label{lem_neighbor_parent_path}
    %At the end of any of the nine steps in \textsc{expand-maxlink} in
    %At the end of any step in any round,
    For any vertices $v$ and $w$, if $w$ is a neighbor of $v$ in the current graph or $w = v.p$, then $v$ and $w$ are in the same component.
    If $v$ is incident with an edge, then there exists a path in the current graph between $v$ and $v.p$.
\end{lemma}
\begin{proof}
    By inspecting Vanilla algorithm (called within \textsc{compact}) and monotonicity, Lemma~\ref{lem_neighbor_parent_path} holds at the beginning of the first round.
    By an induction on rounds, it suffices to prove the following:
    (\romannumeral1) the \textsc{alter} (cf. Steps (\ref{step_1},\ref{step_6})) and Step~(\ref{step_5}) preserve the invariant ($v$ and $w$ are in the same component), as only these two change the neighbor sets;
    (\romannumeral2) the \textsc{maxlink} and \textsc{shortcut} (cf. Steps (\ref{step_1},\ref{step_6})) preserve the invariant, as only these two change the parents.

    Each iteration of the \textsc{maxlink} preserves the invariant since both $v$ and its neighbor $w$, and $w, w.p$ are in the same component by the induction hypothesis. So \textsc{maxlink} preserves the invariant.
    To prove that an \textsc{alter} preserves the invariant, let $(v', w')$ be an edge before \textsc{alter} such that $v'.p = v$ and $w'.p = w$ (to make $w$ a neighbor of $v$ after an \textsc{alter}).
    By the induction hypothesis, $w', w$, $v', v$, and $v', w'$ are in the same component, thus so do $v, w$.
    A \textsc{shortcut} preserves the invariant since $v, v.p$ and $v.p, v.p.p$ are both in the same component.
    Step~(\ref{step_5}) preserves the invariant, because $H(v)$, $H(w)$ are subsets of the neighbor sets of $v$ and $w$ respectively before Step~(\ref{step_5}) (cf. Step~(\ref{step_3})), thus $v, w$ and $w, u$ are both in the same component by the induction hypothesis, and so do $v$ and its new neighbor $u$.
    This finishes the induction and proves the first part of the lemma.

    Now we prove the second part of the lemma.
    %The invariant trivially holds if $v$ is a root.
    By an induction, we assume the invariant holds before some step and prove that each step of \textsc{expand-maxlink} preserves the invariant.
    Before a \textsc{shortcut}, there are paths between $v, v.p$ and $v.p, v.p.p$, thus there is a path between $v$ and $v.p$ after updating $v.p$ to the old $v.p.p$.
    Before an \textsc{alter}, there must exist a child $v'$ of $v$ such that $v'$ is incident with an edge, otherwise $v$ has no incident edge after the \textsc{alter}.
    By the induction hypothesis, there is a path between $v'$ and $v$, which is altered to be a new path between $v$ and $v.p$ after the \textsc{alter}.
    In each iteration of a \textsc{maxlink}, if $v.p$ does not change then the invariant trivially holds;
    otherwise, let $(v, w)$ be the edge that updates $v.p$ to $w.p$.
    By the induction hypothesis, there is a path between $w$ and $w.p$, so there is a path between $v$ and $v.p = w.p$ after the iteration.
    Other steps in the \textsc{expand-maxlink} can only add edges, giving the lemma.
\end{proof}

\begin{proof}[Proof of Lemma~\ref{lem_tree_invariant}]
    Conditions~(\ref{condition_1}) and (\ref{condition_2}) hold at the beginning of the first round by Lemma~\ref{all_tree_flat} and monotonicity of Vanilla algorithm.
    By Lemma~\ref{lem_non_root}, the levels along any tree path from leaf to root are monotonically increasing, so there cannot be a cycle and the labeled digraph must be a collection of trees.
    By Lemma~\ref{lem_neighbor_parent_path} and an induction on tree, all vertices of a tree belong to the same component.
    This proves Condition~(\ref{condition_1}).

    For any tree $T$, if there is an edge between a vertex $v$ in $T$ and another vertex $w$ in another tree $T'$, then
    by Lemma~\ref{lem_neighbor_parent_path}, $v, w$ belong to the same component and so do vertices in $T$ and $T'$, which means $T$ does not contain all the vertices of this component.
    On the other hand, assume $T$ does not contain all the vertices in its component, then there must be another tree $T'$ corresponding to the same component.
    By an induction, we assume the invariant that there is an edge between a vertex $v$ in $T$ and a vertex $w$ in another tree holds, then proves that an \textsc{alter} and (each iteration of) a \textsc{maxlink} preserve the invariant.
    An \textsc{alter} moves the edge between $v$ and $w$ to their parents in their trees, whose two endpoints must belong to different trees as $w$ not in $T$, so the invariant holds.
    Assume a parent update in an iteration of \textsc{maxlink} changes the parent of $u$ to a vertex $u.p$ in $T'$. Let $u'$ be the old parent of $u$ in $T$.
    By Lemma~\ref{lem_neighbor_parent_path}, there are paths between $u, u'$ and $u, u.p$, so there must be a path between $u'$ and $u.p$.
    Since $u'$ and $u.p$ belong to different trees: the new $T$ and $T'$ respectively, there must be en edge from a vertex in $T$ to another tree, and an edge from a vertex in $T'$ to another tree, giving the lemma.
\end{proof}

\begin{lemma} \label{lem_alg_correctness}
    If Faster Connected Components algorithm ends, then it correctly computes the connected components of the input graph.
\end{lemma}
\begin{proof}
    When the repeat loop ends, all trees are flat, which means in the last round, all trees are flat before the \textsc{alter} in Step~(\ref{step_6}) since an \textsc{alter} does not change the labeled digraph. After the \textsc{alter}, all edges are only incident on roots.
    By Lemma~\ref{lem_tree_invariant}, these trees partition the connected components; moreover, if a root is not the only root in its component then there must be an edge between it and another root in its component, i.e., the preconditions of Lemma~\ref{all_tree_flat} is satisfied.
    This implies that the following Connected Components algorithm is applicable and must correctly compute the connected components of the input graph by \S{\ref{sec_cc_alg}}.
\end{proof}

\subsection{Implementation} \label{sec_imp}

In this section, we show how to implement \textsc{expand-maxlink} such that any of the first $O(\log n)$ rounds runs in constant time with good probability.

\begin{lemma} \label{lem_round_constant_time}
    With good probability, any of the first $O(\log n)$ rounds can be implemented to run in $O(1)$ time.
\end{lemma}
\begin{proof}

The \textsc{alter} (cf. Steps (\ref{step_1},\ref{step_6})) applies to all edges in the current graph.
Since each edge corresponds to a distinct processor, Step~(\ref{step_3}) and the \textsc{alter} take $O(1)$ time.

Steps (\ref{step_2},\ref{step_4},\ref{step_7}) and \textsc{shortcut} take $O(1)$ time as each vertex has a corresponding processor and a collision can be detected using the same hash function to write to the same location again.

In each of the $2$ iterations in \textsc{maxlink}, each vertex $v$ updates its parent to a neighbor parent with the highest level if this level is higher than $\ell(v)$.
Since a vertex can increase its level by at most $1$ in any round (cf. Steps (\ref{step_2},\ref{step_7})), there are $O(\log n)$ different levels.
Let each neighbor of $v$ write its parent with level $\ell$ to the $\ell$-th cell of an array of length $O(\log n)$ in the block of $v$ (arbitrary win). Furthermore, by the definitions of level and budget, the block of any vertex with positive level in any round has size at least $b_1 = \Omega(\log^3 n)$ when $c \ge 10$. (Vertex with level $0$ does not have neighbors.)
Therefore, we can assign a processor to each pair of the cells in this array such that each non-empty cell can determine whether there is a non-empty cell with larger index then finds the non-empty cell with largest index in $O(1)$ time,
which contains a vertex with the maximum level.
As a result, Steps (\ref{step_1},\ref{step_6}) take $O(1)$ time.

By Step~(\ref{step_3}), any $w \in H(v)$ has $b(w) = b(v)$, so each $u \in H(w)$ such that $w \in H(v)$ owns a processor in the block of $v$ since $\sqrt{b(v)} \cdot \sqrt{b(w)} = b(v)$ and any vertex in a table is indexed (by its hash value).
Therefore, together with collision detection, Step~(\ref{step_5}) takes $O(1)$ time.

In Step~(\ref{step_8}), each vertex is assigned to a block.
The pool of $\Theta(m)$ processors is partitioned into $\Theta(\log^2 n)$ zones such that the processor allocation in round $r$ for vertices with level $\ell$ uses the zone indexed by $(r, \ell)$, where $r, \ell \in O(\log n)$ in the first $O(\log n)$ rounds. Since there are $\Theta(m)$ processors in total and all the vertex ids are in $[2 m / \log^c n]$ with good probability (cf. Lemma~\ref{lem_stage_1}), we can use $\Theta(m / \log n)$ processors for each different level and apply Lemma~\ref{lem_apx_compaction} to index each root in $O(1)$ time with probability $1 - 1/\text{poly}(n)$ such that the indices of vertices with the same level are distinct, then assign each of them a distinct block in the corresponding zone.
Therefore, Step~(\ref{step_8}) takes $O(1)$ time with good probability by a union bound over all $O(\log n)$ levels and rounds.

Finally, we need to implement the break condition in $O(1)$ time, i.e., to determine whether the graph has diameter $O(1)$ and all trees are flat at the end of each round.
The algorithm checks the following two conditions in each round:
(\romannumeral1) all vertices do not change their parents nor labels in this round, and
(\romannumeral2) for any vertices $v, w, u$ such that $w \in H(v)$, $u \in H(w)$ before Step~(\ref{step_5}), the $h(u)$-th cell in $H(v)$ already contains $u$.
Conditions (\romannumeral1) and (\romannumeral2) can be checked in $O(1)$ time by writing a flag to vertex processor $v$ if they do not hold for $v$, then let each vertex with a flag write the flag to a fixed processor.
If there is no such flag then both Conditions (\romannumeral1) and (\romannumeral2) hold and the loop breaks.
If there is a non-flat tree, some parent must change in the \textsc{shortcut} in Step~(\ref{step_6}).
If all trees are flat, they must be flat before the \textsc{alter} in Step~(\ref{step_6}), then an \textsc{alter} moves all edges to the roots.
Therefore, if Condition (\romannumeral1) holds, all trees are flat and edges are only incident on roots.
Moreover, no level changing means no vertex increase its level in Step~(\ref{step_2}) and there is no dormant vertex in Step~(\ref{step_7}).
So for each root $v$, $N(v) = H(v)$ after Step~(\ref{step_3}) and $N(N(v)) = H(v)$ after Step~(\ref{step_5}) as there is no collision.
By Condition (\romannumeral2), the table $H(v)$ does not change during Step~(\ref{step_5}), so $N(v) = N(N(v))$.
If there exists root $v$ such that there is another root with distance at least $2$ from $v$, then there must exist a vertex $w \neq v$ at distance exactly $2$ from $v$, so $w \notin N(v)$ and $w \in N(N(v))$, contradicting with $N(v) = N(N(v))$.
Therefore, any root is within distance at most $1$ from all other roots in its component and the graph has diameter $O(1)$.

Since each of them runs in $O(1)$ time with good probability, the lemma follows.
\end{proof}

\subsection{Number of Processors} \label{sec_space}

In this section, we show that with good probability, the first $O(\log n)$ rounds use $O(m)$ processors in total.

First of all,
observe that in the case that a root $v$ with level $\ell$ increases its level in Step~(\ref{step_2}) but becomes a non-root at the end of the round, $v$ is not assigned a block of size $b_{\ell(v)}$ in Step~(\ref{step_8}).
Instead, $v$ owns a block of size $b_{\ell} = b_{\ell(v) - 1}$ from the previous round.
Since in later rounds a non-root never participates in obtaining more neighbors by maintaining its table in Steps (\ref{step_3}-\ref{step_5}) (which is the only place that requires a larger block),
such \emph{flexibility} in the relationship between level and budget is acceptable.

By the fact that any root $v$ at the end of any round owns a block of size $b_{\ell(v)} = b_1^{1.01^{\ell(v) - 1}}$,
a non-root can no longer change its level nor budget (cf. Lemma~\ref{lem_non_root}),
and the discussion above,
we obtain:
\begin{corollary} \label{cor_level_budget}
    Any vertex $v$ owns a block of size $b$ at the end of any round where $b_{\ell(v)-1} = b_1^{1.01^{\ell(v) - 2}} \le b \le b_1^{1.01^{\ell(v) - 1}} = b_{\ell(v)}$;
    if $v$ is a root, then the upper bound on $b$ is tight.
\end{corollary}

Secondly, we prove two simple facts related to \textsc{maxlink}.

\begin{lemma} \label{lem_maxlink}
    For any vertex $v$ with parent $v'$ and any $w \in N(v)$ before an iteration of \textsc{maxlink}, it must be $\ell(w.p) \ge \ell(v')$ after the iteration;
    furthermore, if $\ell(w.p) > \ell(v)$ before an iteration, then $v$ must be a non-root after the iteration.
\end{lemma}
\begin{proof}
    For any $w \in N(v)$, its parent has level $\max_{u \in N(w)} \ell(u.p) \ge \ell(v')$ after an iteration of \textsc{maxlink}.
    This implies that if $\ell(w.p) > \ell(v)$ before an iteration, then after that $\ell(v.p)$ is at least the level of the old parent of $w$ which is strictly higher than $\ell(v)$, so $v$ must be a non-root.
\end{proof}

%Secondly, we show that if there are many vertices within distance $2$ from a root with the same budget, then this root is very unlikely to be a root at the end of the round and thus does not increase its budget.
%does not increase its level before expanding its table (cf. Steps (\ref{step_3}-\ref{step_5})), then its neighbors do not have many neighbors to be hashed.

\begin{lemma} \label{lem_many_neighbors}
    For any root $v$ with budget $b$ at the beginning of any round,
    %if $v$ does not increase level in Step~(\ref{step_2}) and is still a root at the end of the round, then with probability $1 - n^{-4}$, any
    if there is a root $w \in N(v)$ with at least $b^{0.1}$ neighbor roots with budget $b$ after Step~(\ref{step_1}), then $v$ either increases level in Step~(\ref{step_2}) or is a non-root at the end of the round with probability $1 - b^{-0.06}$.
\end{lemma}
\begin{proof}
    %Assume for (high-probability) contradiction that there is a $w \in N(v)$ such that there are at least $b^{0.1}$ roots with budget $b$ in $N(w)$.
    Assume $v$ does not increase level in Step~(\ref{step_2}).
    Let $w$ be any root in $N(v)$ after Step~(\ref{step_1}).
    Since each root $u \in N(w)$ with budget $b$ (thus level at least $\ell(v)$ by Corollary~\ref{cor_level_budget}) increases its level with probability $b^{-0.06}$ independently, with probability at least 
    $$ 1 - (1 - b^{-0.06})^{b^{0.1}} \ge 1 - \exp(-b^{0.04}) \ge 1 - n^{-b^{0.03}} \ge 1 - b^{-0.06} , $$
    at least one $u$ increases level to at least $\ell(v) + 1$ in Step~(\ref{step_2}), where in the second inequality we have used the fact that $b^{0.01} \ge b_1^{0.01} \ge \log n$ because $b_1 = \max\{m / n, \log^c n\} / \log^2 n$ and $c \ge 200$.
    Since $u \in N(N(v))$ after Step~(\ref{step_1}), by Lemma~\ref{lem_maxlink}, there is a $w' \in N(v)$ such that $\ell(w'.p) \ge \ell(v) + 1$ after the first iteration of \textsc{maxlink} in Step~(\ref{step_6}).
    Again by Lemma~\ref{lem_maxlink}, this implies that $v$ cannot be a root after the second iteration and the following \textsc{shortcut}.
\end{proof}

Using the above result, we can prove the following key lemma, leading to the total number of processors.

\begin{lemma} \label{lem_increase_level}
    For any root $v$ with budget $b$ at the beginning of any round,
    %if $v$ is still a root at the end of the round, then
    $b(v)$ is increased to $b^{1.01}$ in this round with probability at most $b^{-0.05}$.
\end{lemma}
\begin{proof}
	In Step~(\ref{step_2}), $\ell(v)$ increases with probability $b^{-0.06}$. % when $c \ge 100$, since $b \ge b_1 \ge \log^{c-2} n$.
    If $\ell(v)$ does increase here then it cannot increase again in Step~(\ref{step_7}), so we assume this is not the case (and apply a union bound at the end).

    If there is a root $w \in N(v)$ with at least $b^{0.1}$ neighbor roots with budget $b$ after Step~(\ref{step_1}), then $v$ is a root at the end of the round with probability at most $b^{-0.06}$ by the assumption and Lemma~\ref{lem_many_neighbors}. So we assume this is not the case.

    %If $v$ is a non-root at the end of the round, then its level cannot increase in Step~(\ref{step_7}), so assume $v$ is still a root at the end of the round, which means $v$ is a root during this round by Lemma~\ref{lem_non_root}, .
    %By the statement in Lemma~\ref{lem_many_neighbors} (recall that $v \in N(v)$),
    By the previous assumption we know that at most $b^{0.1}$ vertices are hashed into $H(v)$ in Step~(\ref{step_3}).
    By pairwise independency, with probability at most $(b^{0.1})^2 / \sqrt{b} = b^{-0.3}$ there is a collision as the table has size $\sqrt{b}$, which will increase $\ell(v)$ (cf. Steps (\ref{step_4},\ref{step_7})). %and its budget in Step~(\ref{step_8}).

    Now we assume that there is no collision in $H(v)$ in Step~(\ref{step_3}), which means $H(v)$ contains all the at most $b^{0.1}$ neighbor roots with budget $b$.
    By the same assumption, each such neighbor root $w$ has at most $b^{0.1}$ neighbor roots with budget $b$, so there is a collision in $H(w)$ in Step~(\ref{step_3}) with probability at most $(b^{0.1})^2 / \sqrt{b} = b^{-0.3}$. By a union bound over all the $|H(v)| \le b^{0.1}$ such vertices, $v$ is marked as dormant in the second statement of Step~(\ref{step_4}) (and will increase level in Step~(\ref{step_7})) with probability $b^{-0.2}$.

    It remains to assume that there is no collision in $H(v)$ nor in any $H(w)$ such that $w \in H(v)$ after Step~(\ref{step_4}).
    As each such table contains at most $b^{0.1}$ vertices, in Step~(\ref{step_5}) there are at most $b^{0.2}$ vertices to be hashed,
    resulting in a collision in $H(v)$ with probability at most $(b^{0.2})^2 / \sqrt{b} = b^{-0.1}$, which increases $\ell(v)$ in Step~(\ref{step_7}).

	Observe that only a root $v$ at the end of the round can increase its budget, and the increased budget must be $b^{1.01}$ since the level can increase by at most $1$ during the round and $b = b_{\ell(v)}$ at the beginning of the round by Corollary~\ref{cor_level_budget}.
    By a union bound over the events in each paragraph, $b(v)$ is increased to $b^{1.01}$ with probability at most $b^{-0.06} + b^{-0.06} + b^{-0.3} + b^{-0.2} + b^{-0.1} \le b^{-0.05}$.
\end{proof}

Finally, we are ready to prove an upper bound on the number of processors.

\begin{lemma} \label{lem_processor_number}
    With good probability, the first $O(\log n)$ rounds use $O(m)$ processors in total.
\end{lemma}
\begin{proof}
	%Using Lemma~\ref{lem_increase_level}, by a union bound over all $O(n)$ roots, all $O(\log n)$ rounds, and all $O(\log n)$ different budgets (since there are $O(\log n)$ different levels), with probability at least $1 - n^{-3}$, any root $v$ with budget $b$ at the beginning of any round increases its budget to $b^{1.01}$ with probability at most $b^{-0.05}$. We may assume that the $(1 - n^{-3})$-probability event always holds as a good-probability result follows from a union bound.
    The number of processors for (altered) original edges and vertices are clearly $O(m)$ over all rounds (where each vertex processor needs $O(1)$ private memory to store the corresponding parent, (renamed) vertex id, the two words for pairwise independent hash function, level, and budget).
    By the proof of Lemma~\ref{lem_round_constant_time}, with good probability the processor allocation in Step~(\ref{step_8}) always succeeds and has a multiplicative factor of $2$ in the number of processors allocated before mapping the indexed vertices to blocks (cf. Definition~\ref{def_apx_compaction}).
    Therefore, we only need to bound the number of processors in blocks that are assigned to a vertex in Step~(\ref{step_8}) in all $O(\log n)$ rounds.

    For any positive integer $\ell$, let $n_{\ell}$ be the number of vertices that ever reach budget $b_{\ell}$ during the first $O(\log n)$ rounds.
    For any vertex $v$ that ever reaches budget $b_{\ell}$, it has exactly one chance to reach budget $b_{\ell + 1}$ in a round if $v$ is a root in that round, which happens with probability at most ${b_{\ell}}^{-0.05}$ by Lemma~\ref{lem_increase_level}.
    By a union bound over all $O(\log n)$ rounds, $v$ reaches budget $b_{\ell + 1}$ with probability at most $O(\log n) \cdot {b_{\ell}}^{-0.05} \le {b_{\ell}}^{-0.04}$ when $c \ge 200$, since $b_{\ell} \ge b_1 \ge \log^{c-2} n$.
    We obtain $\E[n_{\ell + 1} \mid n_{\ell}] \le n_{\ell} \cdot {b_{\ell}}^{-0.04}$, thus by $b_{\ell+1} = {b_{\ell}}^{1.01}$, it must be:
    \begin{equation*}
        \E[n_{\ell + 1} b_{\ell + 1} \mid n_{\ell}] \le n_{\ell} \cdot {b_{\ell}}^{-0.04} \cdot {b_{\ell}}^{1.01} = n_{\ell} b_{\ell} \cdot {b_{\ell}}^{-0.03}.
    \end{equation*}

    By Markov's inequality, $n_{\ell + 1} b_{\ell + 1} \le n_{\ell} b_{\ell}$ with probability at least $1 - {b_{\ell}}^{-0.03} \ge 1 - {b_1}^{-0.03}$. %, which implies $n_{\ell + 1} b_{\ell + 1} \le n_{\ell} b_{\ell}$ with probability at least $1 - {b_1}^{-0.02}$.
    By a union bound on all $\ell \in O(\log n)$, $n_{\ell} b_{\ell} \le n_1 b_1$ for all $\ell \in O(\log n)$ with probability at least $1 - O(\log n) \cdot {b_1}^{-0.03} \ge 1 - {b_1}^{-0.01}$, which is $1 - 1 / \text{poly}((m\log n) / n)$ by $b_1 = \max\{m / n, \log^c n\} / \log^2 n$ and $c \ge 100$.
    %If this happens, we obtain $\sum_{\ell} n_{\ell}  b_{\ell} \le \sum_{\ell} n_1 b_1 \cdot 2^{-\ell + 1} \le 2 n_1 b_1$.
    %In any of the first $O(\log n)$ rounds, any vertex that owns a block of size $b_{\ell}$ is at level at least $\ell$ (cf. Corollary~\ref{cor_level_budget}), so the number of such vertices in this round is at most $n_{\ell}$.
    So the number of new allocated processors for vertices with any budget in any of the first $O(\log n)$ rounds is at most $n_1 b_1$ with good probability.

    Recall from Lemma~\ref{lem_stage_1} that with good probability,
    \begin{align*}
        n_1 \cdot b_1 &= \min\{n, 2m / \log^c n\} \cdot \max\{m/n, \log^c n\} / \log^2 n \\
        &= O(m / \log^2 n) .
    \end{align*}
    Finally, by a union bound over all the $O(\log n)$ levels and $O(\log n)$ rounds, with good probability the total number of processors is $O(m)$.
\end{proof}

\subsection{Diameter Reduction} \label{sec_diameter_reduction}

Let $R \coloneqq O(\log d + \log\log_{m/n} n)$ where the constant hidden in $O(\cdot)$ will be determined later in this section.
The goal is to prove that $O(R)$ rounds of \textsc{expand-maxlink} suffice to reduce the diameter of the graph to $O(1)$ and flatten all trees with good probability.

In a high level, our algorithm/proof is divided into the following $3$ stages/lemmas:
\begin{lemma} \label{lem_dr1}
    With good probability, after round $R$, the diameter of the graph is $O(R)$.
\end{lemma}
\begin{lemma} \label{lem_dr2}
    With good probability, after round $O(R)$, the diameter of the graph is at most $1$.
\end{lemma}
\begin{lemma} \label{lem_flatten}
    With good probability, after round $O(R)$, the diameter of the graph is $O(1)$ and all trees are flat.
\end{lemma}

Lemma~\ref{lem_flatten} implies that the repeat loop must end in $O(R)$ rounds with good probability by the break condition. %, and the algorithm is correct (cf. Lemma~\ref{lem_alg_correctness}).

\subsubsection{Path Construction} \label{sec_path_construction}

To formalize and quantify the effect of reducing the diameter,
consider any shortest path $P$ in the original input graph $G$.
When (possibly) running Vanilla algorithm on $G$ in \textsc{compact}, an \textsc{alter} replaces each vertex on $P$ by its parent, resulting in a path $P'$ of the same length as $P$.
%\footnote{Denote length (the number of edges on $P$) by $|P|$. For any $i \in [|P| + 1]$, let $P(i)$ be the $i$-th vertex on $P$.}
(Note that $P'$ might not be a shortest path in the current graph and can contain loops.)
Each subsequent \textsc{alter} in \textsc{compact}
continues to replace the vertices on the old path by their parents to get a new path.
Define $P_1$ to be the path obtained from $P$ by the above process at the beginning of round $1$, by $|P| \le d$ we immediately get:
\begin{corollary}\label{cor_p1}
    For any $P_1$ corresponding to a shortest path in $G$, $|P_1| \le d$.
\end{corollary}

In the following, we shall fix a shortest path $P$ in the original graph and consider its corresponding paths over rounds.
Each \textsc{alter} (cf. Steps~(\ref{step_1},\ref{step_6})) in each round does such replacements to the old paths, but we also add edges to the graph for reducing the diameter of the current graph:
for any vertices $v$ and $w$ on path $P'$, if the current graph contains edge $(v, w)$, then all vertices exclusively between $v$ and $w$ can be removed from $P'$, which still results in a valid path in the current graph from the first to the last vertex of $P'$, reducing the length.
If all such paths reduce their lengths to at most $d'$, the diameter of the current graph is at most $d'$.
Formally, we have the following inductive construction of paths for diameter reduction:
\begin{definition}[path construction] \label{def_path_construction}
    Let all vertices on $P_1$ be \emph{active}.
    For any positive integer $r$, given path $P_{r}$ with at least $4$ active vertices at the beginning of round $r$, \textsc{expand-maxlink} constructs $P_{r+1}$ by the following $7$ phases:
    \begin{enumerate}
        \item The \textsc{alter} in Step~(\ref{step_1}) replaces each vertex $v$ on $P_{r}$ by $v' \coloneqq v.p$ to get path $P_{r,1}$. For any $v'$ on $P_{r,1}$, let $\underline{v'}$ be on $P_{r}$ such that $\underline{v'}.p = v'$.\footnote{Semantically, there might be more than one $\underline{v'}$ with the same parent $v'$, but our reference always comes with an index on the replaced path, whose corresponding position in the old path is unique. All subsequent names follow this manner.} \label{phase_1}
        \item Let the subpath containing all active vertices on $P_{r,1}$ be $P_{r,2}$. \label{phase_2}
        \item After Step~(\ref{step_5}), set $i$ as $1$, and repeat the following until $i \ge |P_{r,2}| - 1$: let $v' \coloneqq P_{r,2}(i)$,
            if $\underline{v'}$ is a root and does not increase level during round $r$ then: if the current graph contains edge $(v', P_{r,2}(i+2))$ then mark $P_{r,2}(i+1)$ as \emph{skipped} and set $i$ as $i+2$;
            else set $i$ as $i+1$. \label{phase_3}
        \item For each $j \in [i + 1, |P_{r,2}| + 1]$, mark $P_{r,2}(j)$ as \emph{passive}. \label{phase_4}
        \item Remove all skipped and passive vertices from $P_{r,2}$ to get path $P_{r,5}$. \label{phase_5}
        \item Concatenate $P_{r,5}$ with all passive vertices on $P_{r,1}$ and $P_{r,2}$ to get path $P_{r,6}$. \label{phase_6}
        \item The \textsc{alter} in Step~(\ref{step_6}) replaces each vertex $v$ on $P_{r,6}$ by $v.p$ to get path $P_{r+1}$. \label{phase_7}
    \end{enumerate}
    For any vertex $v$ on $P_r$ that is replaced by $v'$ in Phase~(\ref{phase_1}), if $v'$ is not skipped in Phase~(\ref{phase_3}), then let $\overline{v}$ be the vertex replacing $v'$ in Phase~(\ref{phase_7}), and call $\overline{v}$ the \emph{corresponding vertex} of $v$ in round $r+1$.
\end{definition}

\begin{lemma} \label{lem_valid_path}
    For any non-negative integer $r$, the $P_{r+1}$ constructed in Definition~\ref{def_path_construction} is a valid path in the graph and all passive vertices are consecutive from the successor of the last active vertex to the end of $P_{r+1}$.
\end{lemma}
\begin{proof}
    The proof is by an induction on $r$.
    Initially, $P_1$ is a valid path by our discussion on \textsc{alter} at the beginning of this section: it only replaces edges by new edges in the altered graph; moreover, the second part of the lemma is trivially true as all vertices are active.
    Assuming $P_{r}$ is a valid path and all passive vertices are consecutive from the successor of the last active vertex to the end of the path.
    We show the inductive step by proving the invariant after each of the $7$ phases in Definition~\ref{def_path_construction}.
    Phase~(\ref{phase_1}) maintains the invariant.
    In Phase~(\ref{phase_2}), $P_{r,2}$ is a valid path as all active vertices are consecutive at the beginning of $P_{r,1}$ (induction hypothesis).
    In Phase~(\ref{phase_3}), if a vertex $v$ is skipped, then there is an edge between its predecessor and successor on the path; otherwise there is an edge between $v$ and its successor by the induction hypothesis;
    all passive vertices are consecutive from the successor of the last non-skipped vertex to the end of $P_{r,2}$ (cf. Phase~(\ref{phase_4})), so the invariant holds.
    In Phase~(\ref{phase_6}), since the first passive vertex on $P_{r,2}$ is a successor of the last vertex on $P_{r,5}$ and the last passive vertex on $P_{r,2}$ is a predecessor of the first passive vertex on $P_{r,1}$ (induction hypothesis), the invariant holds.
    Phase~(\ref{phase_7}) maintains the invariant.
    Therefore, $P_{r+1}$ is a valid path and all passive vertices are consecutive from the successor of the last active vertex to the end of $P_{r+1}$.
\end{proof}

Now we relate the path construction to the diameter of the graph:
\begin{lemma} \label{lem_path_diameter}
    For any positive integer $r$, the diameter of the graph at the end of round $r$ is $O(\max_{P_r} |P_{r,2}| + r)$.
\end{lemma}
\begin{proof}
    %In any round of the algorithm, if there is a path $P$ between vertices $v$ and $w$, then the distance between $v$ and $w$ is at most $|P|$.  If all paths in the graph have length at most $\delta$, then the diameter is at most $\delta$.
    %For any component in the original input graph, consider its shortest-path tree $T$, which is a spanning tree.
    Any shortest path $P$ in the original input graph from $s$ to $t$ is transformed to the corresponding $P_1$ at the beginning of round $1$, which maintains connectivity between the corresponding vertices of $s$ and $t$ on $P_1$ respectively. %, where the corresponding vertex of any vertex $v$ after an \textsc{alter} is the parent of the corresponding vertex before that \textsc{alter}.
    By an induction on the number of \textsc{alter}s and Lemma~\ref{lem_valid_path}, the corresponding vertices of $s$ and $t$ are still connected by path $P_{r+1}$ at the end of round $r$.
    Note that by Lemma~\ref{lem_valid_path}, $P_{r+1}$ can be partitioned into two parts after Phase~(\ref{phase_2}): subpath $P_{r,2}$ and the subpath containing only passive vertices.
    Since in each round we mark at most $2$ new passive vertices (cf. Phases (\ref{phase_3},\ref{phase_4})),
    we get $|P_{r+1}| \le |P_{r,5}| + 2r \le |P_{r,2}| + 2r$.
    If any path $P_{r+1}$ that corresponds to a shortest path in the original graph have length at most $d'$, the graph at the end of round $r$ must have diameter at most $d'$, so the lemma follows.
\end{proof}

It remains to bound the length of any $P_{r,2}$ in any round $r$, which relies on the following potential function:
\begin{definition}\label{def_path_potential}
    For any vertex $v$ on $P_{1}$, define its \emph{potential} $\phi_1(v) \coloneqq 1$.
    For any positive integer $r$, given path $P_r$ with at least $4$ active vertices at the beginning of round $r$ and the potentials of vertices on $P_r$, define the \emph{potential} of each vertex on $P_{r+1}$ based on Definition~\ref{def_path_construction} as follows:\footnote{The second subscript of a potential indicates the phase to obtain that potential; the subscript for paths in Definition~\ref{def_path_construction} follows the same manner.}
    \begin{itemize}
        \item For each $v$ replaced by $v.p$ in Phase~(\ref{phase_1}), $\phi_{r, 1}(v.p) \coloneqq \phi_r(v)$.
        \item After Phase~(\ref{phase_4}), for each active vertex $v$ on $P_{r,2}$, if the successor $w$ of $v$ is skipped or passive, then $\phi_{r,4}(v) \coloneqq \phi_{r,1}(v) + \phi_{r,1}(w)$.
        \item After Phase~(\ref{phase_6}), for each vertex $v$ on $P_{r,6}$, if $v$ is active on $P_{r,2}$, then $\phi_{r,6}(v) \coloneqq \phi_{r,4}(v)$, otherwise $\phi_{r,6}(v) \coloneqq \phi_{r,1}(v)$.
        \item For each $v$ replaced by $v.p$ in Phase~(\ref{phase_7}), $\phi_{r+1}(v.p) \coloneqq \phi_{r,6}(v)$.
    \end{itemize}
\end{definition}

We conclude this section by some useful properties of potentials.
\begin{lemma} \label{lem_potential_property}
    For any path $P_r$ at the beginning of round $r \ge 1$, the following holds:
    (\romannumeral1) $\sum_{v \in P_r} \phi_r(v) \le d+1$;
    (\romannumeral2) for any $v$ on $P_r$, $\phi_r(v) \ge 1$;
    (\romannumeral3) for any non-skipped $v$ on $P_{r,2}$ and its corresponding vertex $\overline{v}$ on $P_{r+1}$, $\phi_{r+1}(\overline{v}) \ge \phi_r(v)$.
\end{lemma}
\begin{proof}
    The proof is by an induction on $r$.
    The base case follows from $\phi(v) = 1$ for each $v$ on $P_1$ (cf. Definition~\ref{def_path_potential}) and Corollary~\ref{cor_p1}.
    For the inductive step, note that by Definition~\ref{def_path_potential}, the potential of a corresponding vertex is at least the potential of the corresponding vertex in the previous round (and can be larger in the case that its successor is skipped or passive). This gives (\romannumeral2) and (\romannumeral3) of the lemma.
    %For the first part of the lemma, let $P_{r+1}$ be the corresponding path of $P_r$ after round $r$.
    For any vertex $u$ on $P_{r,2}$, if both $u$ and its successor are active, then $\phi_r(u)$ is presented for exactly $1$ time in $\sum_{v \in P_r} \phi_r(v)$ and $\sum_{v \in P_{r+1}} \phi_{r+1}(v)$ respectively;
    if $u$ is active but its successor $w$ is skipped or passive, then $\phi_r(u) + \phi_r(w)$ is presented for exactly $1$ time in each summations as well;
    if $u$ and its predecessor are both passive, then $\phi_r(u)$ is presented only in $\sum_{v \in P_r} \phi_r(v)$;
    the potential of the last vertex on $P_{r,2}$ might not be presented in $\sum_{v \in P_{r+1}} \phi_{r+1}(v)$ depending on $i$ after Phase~(\ref{phase_3}).
    Therefore, $\sum_{v \in P_{r+1}} \phi_{r+1}(v) \le \sum_{v \in P_r} \phi_r(v)$ and the lemma holds.
\end{proof}

\subsubsection{Proof of Lemma~\ref{lem_dr1}} \label{lem_pf_dr1}

Now we prove Lemma~\ref{lem_dr1}, which relies on several results.
First of all, we need an upper bound on the maximal possible level:
\begin{lemma} \label{lem_max_level}
    With good probability, the level of any vertex in any of the first $O(\log n)$ rounds is at most $L \coloneqq 1000 \max\{2, \log\log_{m/n} n\}$. %$O(\log\log_{m/n} n)$.
\end{lemma}
\begin{proof}
    By Lemma~\ref{lem_processor_number}, with good probability the total number of processors used in the first $O(\log n)$ rounds is $O(m)$.
    We shall condition on this happening then assume for contradiction that there is a vertex $v$ with level at least $L$ in some round.

    If $\log\log_{m/n} n \le 2$, then $m / n \ge n^{1/4}$.
    By Corollary~\ref{cor_level_budget}, a block owned by $v$ has size at least
    \begin{equation*}
        {b_1}^{1.01^{2000 - 2}} \ge {b_1}^{20} \ge (m / n / \log^2 n)^{20} \ge (n^{1/5})^{20} = n^4 \ge m^2 ,
    \end{equation*}
    which is a contradiction as the size of this block owned by $v$ exceeds the total number of processors $O(m)$.

    Else if $\log\log_{m/n} n > 2$, then by Corollary~\ref{cor_level_budget}, a block owned by $v$ has size at least
    \begin{equation} \label{eq_b_size}
        {b_1}^{1.01^{L - 2}} \ge {b_1}^{1.01^{999 \log\log_{m/n} n}} \ge {b_1}^{(\log_{m/n} n)^{10}} \ge {b_1}^{8 \log_{m/n} n} .
    \end{equation}

    Whether $m /n \le \log^c n$ or not, if $c \ge 10$, it must be $b_1 = \max\{m/n, \log^c n\} / \log^2 n \ge \sqrt{m/n}$.
    So the value of (\ref{eq_b_size}) is at least $n^4 \ge m^2$, contradiction.
    Therefore, the level of any vertex is at most $L$.
\end{proof}

%As outlined in \S{\ref{sec_overview}},
We also require the following key lemma:

\begin{lemma} \label{lem_2hops}
    %With probability $1 - n^{-4}$,
    For any root $v$ and any $u \in N(N(v))$ at the beginning of any round,
    let $u'$ be the parent of $u$ after Step~(\ref{step_1}).
    If $v$ does not increase level and is a root during this round, then $u' \in H(v)$ after Step~(\ref{step_5}).
\end{lemma}

To prove Lemma~\ref{lem_2hops}, we use another crucial property of the algorithm, which is exactly the reason behind the design of \textsc{maxlink}.
\begin{lemma} \label{lem_same_root_level}
    For any root $v$ and any $u \in N(N(v))$ at the beginning of any round, if $v$ does not increase level in Step~(\ref{step_2}) and is a root at the end of the round, then $u.p$ is a root with budget $b(v)$ after Step~(\ref{step_1}).
\end{lemma}
\begin{proof}
    By Lemma~\ref{lem_non_root}, $v$ is a root during this round.
    For any $w \in N(v)$ and any $u \in N(w)$, applying Lemma~\ref{lem_maxlink} for $2$ times, we get that $\ell(v) \le \ell(w.p)$ and $\ell(v) \le \ell(u.p)$ after the \textsc{maxlink} in Step~(\ref{step_1}).
    If there is a $u \in N(N(v))$ such that $u.p$ is a non-root or $\ell(u.p) > \ell(v)$ before the \textsc{alter} in Step~(\ref{step_1}), it must be $\ell(u.p.p) > \ell(v)$ by Lemma~\ref{lem_non_root}.
    Note that $u.p$ is in $N(N(v))$ after the \textsc{alter}, which still holds before Step~(\ref{step_6}) as we only add edges.
    By Lemma~\ref{lem_maxlink}, there is a $w' \in N(v)$ such that $\ell(w'.p) > \ell(v)$ after the first iteration of \textsc{maxlink} in Step~(\ref{step_6}). Again by Lemma~\ref{lem_maxlink}, this implies that $v$ cannot be a root after the second iteration, a contradiction.
    Therefore, for any $u \in N(N(v))$, $u.p$ is a root with level $\ell(v)$ (thus budget $b(v)$) after Step~(\ref{step_1}).
\end{proof}
With the help of Lemma~\ref{lem_same_root_level} we can prove Lemma~\ref{lem_2hops}:
\begin{proof}[Proof of Lemma~\ref{lem_2hops}]
    For any vertex $u$, let $N'(u)$ be the set of neighbors after Step~(\ref{step_1}).
    First of all, we show that after Step~(\ref{step_5}), $H(v)$ contains all roots in $N'(N'(v))$ with budget $b$, where $b$ is the budget of $v$ at the beginning of the round.
    For any root $w \in N'(v)$, in Step~(\ref{step_3}), all roots with budget $b(w)$ in $N'(w)$ are hashed into $H(w)$.
    If there is a collision in any $H(w)$, then $v$ must be dormant (cf. Step~(\ref{step_4})) thus increases level in either Step (\ref{step_2}) or (\ref{step_7}), contradiction.
    So there is no collision in $H(w)$ for any $w \in N'(v)$, which means $H(w) \supseteq N'(w)$.
    Recall that $v \in N'(v)$ and we get that all roots with budget $b(w) = b$ from $N'(N'(v))$ are hashed into $H(v)$ in Step~(\ref{step_5}).
    Again, if there is a collision, then $v$ must be dormant and increase level in this round.
    Therefore, $N'(N'(v)) \subseteq H(v)$ at the end of Step~(\ref{step_5}).

    By Lemma~\ref{lem_same_root_level}, for any $u \in N(N(v))$ at the beginning of any round, $u' = u.p$ is a root with budget $b$ in $N'(N'(v))$ after Step~(\ref{step_1}).
    Therefore, $u' \in H(v)$ at the end of Step~(\ref{step_5}), giving Lemma~\ref{lem_2hops}.
\end{proof}

The proof of Lemma~\ref{lem_dr1} relies on the following lemma based on potentials:
\begin{lemma} \label{lem_potential_lb}
    %With good probability,
    At the beginning of any round $r \ge 1$, for any active vertex $v$ on any path $P_r$, $\phi_r(v) \ge 2^{r - \ell(v)}$.
\end{lemma}
\begin{proof}
    The proof is by an induction on $r$.
    The base case holds because for any (active) vertex $v$ on $P_1$, $\phi_1(r) = 1$ and $r = \ell(v) = 1$ (other vertices do not have incident edges).
    Now we prove the inductive step from $r$ to $r+1$ given that the corresponding vertex $\overline{v}$ of $v \in P_r$ is on $P_{r+1}$ and active.

    Suppose $v$ is a non-root at the end of round $r$.
    If $v$ is a non-root at the end of Step~(\ref{step_1}), then $\ell(v.p) > \ell(v)$ after Step~(\ref{step_1}) by Lemma~\ref{lem_non_root}, and $\ell(\overline{v}) \ge \ell(v.p) > \ell(v)$;
    else if $v$ first becomes a non-root in Step~(\ref{step_6}), then $\overline{v} = v.p$ and $\ell(v.p) > \ell(v)$ after Step~(\ref{step_6}) by Lemma~\ref{lem_non_root}.
    So by the induction hypothesis, $\phi_{r+1}(\overline{v}) \ge \phi_{r}(v) \ge 2^{r - \ell(v)} \ge 2^{r+1 - \ell(\overline{v})}$.

    Suppose $v$ increases its level in round $r$.
    Let $\ell$ be the level of $v$ at the beginning of round $r$.
    If the increase happens in Step~(\ref{step_2}), then $v$ is a root after Step~(\ref{step_1}).
    Whether $v$ changes its parent in Step~(\ref{step_6}) or not, the level of $\overline{v} = v.p$ is at least $\ell+1$.
    Else if the increase happens in Step~(\ref{step_7}), then $v$ is a root after Step~(\ref{step_6}).
    So $\overline{v} = v$ and its level is at least $\ell + 1$ at the end of the round.
    By the induction hypothesis, $\phi_{r+1}(\overline{v}) \ge \phi_{r}(v) \ge 2^{r - \ell} \ge 2^{r+1 - \ell(\overline{v})}$.

    It remains to assume that $v$ is a root and does not increase level during round $r$.
    By Lemma~\ref{lem_2hops}, for any $u \in N(N(v))$ at the beginning of round $r$, the parent $u'$ of $u$ after Step~(\ref{step_1}) is in $H(v)$ after Step~(\ref{step_5}).
    Since $v$ is a root during the round, it remains on $P_{r,2}$ after Phase~(\ref{phase_2}).
    We discuss two cases depending on whether $v$ is at position before $|P_{r,2}| - 1$ or not.

    In Phase~(\ref{phase_3}), note that if $v = P_{r,2}(i)$ where $i < |P_{r,2}| - 1$, then $P_{r,2}(i+2)$ is the parent of a vertex in $N(N(v))$ after Step~(\ref{step_1}), which must be in $H(v)$ after Step~(\ref{step_5}). Therefore, the graph contains edge $(v, P_{r,2}(i+2))$ and $v' \coloneqq P_{r,2}(i+1)$ is skipped, thus $\phi_{r,4}(v) = \phi_{r,1}(v) + \phi_{r,1}(v')$ by Definition~\ref{def_path_potential}.
    Since $i + 1 \le |P_{r,2}| + 1$, $v'$ is an active vertex on $P_{r,1}$.
    By the induction hypothesis, $\phi_{r,1}(v') = \phi_r(\underline{v'}) \ge 2^{r - \ell(\underline{v'})}$ (recall that $\underline{v'}$ is replaced by its parent $v'$ in Phase~(\ref{phase_1})/Step~(\ref{step_1})).
    If $\ell(\underline{v'}) > \ell(v)$, then applying Lemma~\ref{lem_maxlink} for two times we get that $v$ is a non-root after Step~(\ref{step_1}), a contraction.
    Therefore, $\phi_{r,1}(v') \ge 2^{r - \ell(\underline{v'})} \ge 2^{r - \ell(v)}$ and $\phi_{r,4}(v) \ge \phi_{r,1}(v) + \phi_{r,1}(v') \ge \phi_r(v) + 2^{r - \ell(v)} \ge 2^{r + 1 - \ell(v)}$.

    On the other hand, if $i \ge |P_{r,2}| - 1$ is reached after Phase~(\ref{phase_3}), it must be $i < |P_{r,2}| + 1$ by the break condition of the loop in Phase~(\ref{phase_3}).
    Note that $v' \coloneqq P_{r,2}(i+1)$ is marked as passive in Phase~(\ref{phase_4}), and by Definition~\ref{def_path_potential}, $\phi_{r,4} = \phi_{r,1}(v) + \phi_{r,1}(v')$.
    Moreover, since $i + 1 \le |P_{r,2}| + 1$, $v'$ is an active vertex on $P_{r,1}$.
    Using the same argument in the previous paragraph, we obtain $\phi_{r,4}(v) \ge 2^{r + 1 - \ell(v)}$.

    By Definition~\ref{def_path_potential}, after Phase~(\ref{phase_7}), $\phi_{r+1}(\overline{v}) = \phi_{r,6}(v) = \phi_{r,4}(v) \ge 2^{r + 1 - \ell(v)} = 2^{r + 1 - \ell(\overline{v})}$.
    As a result, the lemma holds for any active vertex $\overline{v}$ on $P_{r+1}$, finishing the induction and giving the lemma.
\end{proof}

\begin{proof}[Proof of Lemma~\ref{lem_dr1}]
    Let $R \coloneqq \log d + L$, where $L$ is defined in Lemma~\ref{lem_max_level}.
    By Lemma~\ref{lem_max_level}, with good probability, $\ell(v) \le L$ for any vertex $v$ in any of the first $O(\log n)$ rounds, and we shall condition on this happening.
    By Lemma~\ref{lem_potential_lb}, at the beginning of round $R$, if there is a path $P_R$ of at least $4$ active vertices, then for any of these vertices $v$, it must be $\phi_R(v) \ge 2^{R - \ell(v)} \ge 2^{R - L} \ge d$.
    So $\sum_{v \in P_R} \phi_r(v) \ge 4d > d+1$, contradicting with Lemma~\ref{lem_potential_property}.
    Thus, any path $P_R$ has at most $3$ active vertices, which means $|P_{R,2}| \le 3$ by Definition~\ref{def_path_construction}.
    Therefore, by Lemma~\ref{lem_path_diameter}, the diameter of the graph at the end of round $R$ is $O(R)$ with good probability.
\end{proof}

\subsubsection{Proof of Lemma~\ref{lem_flatten}} \label{sec_pf_flatten}

Now we prove Lemma~\ref{lem_flatten} as outlined at the beginning of \S{\ref{sec_diameter_reduction}}.
Based on the graph and any $P_R$ at the beginning of round $R+1$, we need a (much simpler) path construction:
\begin{definition} \label{def_path_construction_simper}
    For any integer $r > R$, given path $P_{r}$ with $|P_r| \ge 3$ at the beginning of round $r$, \textsc{expand-maxlink} constructs $P_{r+1}$ by the following:
    \begin{enumerate}
        \item The \textsc{alter} in Step~(\ref{step_1}) replaces each vertex $v$ on $P_{r}$ by $v' \coloneqq v.p$ to get path $P_{r,1}$. For any $v'$ on $P_{r,1}$, let $\underline{v'}$ be on $P_{r}$ such that $\underline{v'}.p = v'$.
        \item After Step~(\ref{step_5}), let $v' \coloneqq P_{r,1}(1)$,
            if $\underline{v'}$ is a root at the end of round $r$ and does not increase level during round $r$ then: if the current graph contains edge $(v', P_{r,1}(3))$ then remove $P_{r,1}(2)$ to get path $P_{r,2}$.
        \item The \textsc{alter} in Step~(\ref{step_6}) replaces each vertex $v$ on $P_{r,2}$ by $v.p$ to get path $P_{r+1}$.
    \end{enumerate}
    For any vertex $v$ on $P_r$ that is replaced by $v'$ in the first step, if $v'$ is not removed in the second step, then let $\overline{v}$ be the vertex replacing $v'$ in the third step, and call $\overline{v}$ the \emph{corresponding vertex} of $v$ in round $r+1$.
\end{definition}

An analog of Lemma~\ref{lem_valid_path} immediately shows that $P_r$ is a valid path for any $r \ge R+1$.
The proof of Lemma~\ref{lem_dr2} is simple enough without potential:
\begin{proof}[Proof of Lemma~\ref{lem_dr2}]
    By Lemma~\ref{lem_dr1}, at the beginning of round $R+1$, with good probability, any $P_{R+1}$ has length $O(R)$. We shall condition on this happening and apply a union bound at the end of the proof.
    In any round $r > R$, for any path $P_r$ with $|P_r| \ge 3$, consider the first vertex $v'$ on $P_{r,1}$ (cf. Definition~\ref{def_path_construction_simper}).
    If $\underline{v'}$ is a non-root or increases its level during round $r$, then by the first $3$ paragraphs in the proof of Lemma~\ref{lem_potential_lb}, it must be $\ell(\overline{v'}) \ge \ell(\underline{v'}) + 1$.
    Otherwise, by Lemma~\ref{lem_2hops}, there is an edge between $v'$ and the successor of its successor in the graph after Step~(\ref{step_5}), which means the successor of $v'$ on $P_{r,1}$ is removed in the second step of Definition~\ref{def_path_construction_simper}.
    Therefore, the number of vertices on $P_{r+1}$ is one less than $P_r$ if $\ell(\overline{v'}) = \ell(\underline{v'})$ as the level of a corresponding vertex cannot be lower.
    By Lemma~\ref{lem_max_level}, with good probability, the level of any vertex in any of the $O(\log n)$ rounds cannot be higher than $L$.
    As $P_{R+1}$ has $O(R)$ vertices, in round $r = O(R) + L + R = O(R) \le O(\log n)$, the number of vertices on any $P_r$ is at most $2$.
    Therefore, the diameter of the graph after $O(R)$ rounds is at most $1$ with good probability.
\end{proof}
\begin{proof}[Proof of Lemma~\ref{lem_flatten}]
    Now we show that after the diameter reaches $1$, if the loop has not ended, then the loop must break in $2L + \log_{5/4} L$ rounds with good probability, i.e., the graph has diameter $O(1)$ and all trees are flat.

    For any component, let $u$ be a vertex in it with the maximal level and consider any (labeled) tree of this component.
    For any vertex $v$ in this tree that is incident with an edge, since the diameter is at most $1$, $v$ must have an edge with $u$, which must be a root.
    So $v$ updates its parent to a root with the maximal level after a \textsc{maxlink}, then any root must have the maximal level in its component since a root with a non-maximal level before the \textsc{maxlink} must have an edge to another tree (cf. Lemma~\ref{lem_tree_invariant}).
    Moreover, if $v$ is a root, this can increase the maximal height among all trees in its component by $1$.

    Consider the tree with maximal height $\xi$ in the labeled digraph after Step~(\ref{step_1}).
    By Lemmas \ref{lem_non_root} and \ref{lem_max_level}, with good probability $\xi \le L$.
    The maximal level can increase by at most $1$ in this round.
    If it is increased in Step~(\ref{step_7}), the maximal height is at most
    $\lceil \xi / 2 \rceil + 1$ after the \textsc{maxlink} in the next round;
    otherwise, the maximal height is at most $\lceil (\xi + 1) / 2 \rceil \le \lceil \xi / 2 \rceil + 1$.
    If $\xi \ge 4$, then the maximal height of any tree after Step~(\ref{step_1}) in the next round is at most $(4/5) \xi$ (the worst case is that a tree with height $5$ gets shortcutted to height $3$ in Step~(\ref{step_6}) and increases its height by $1$ in the \textsc{maxlink} of Step~(\ref{step_1}) in the next round).
    Therefore, after $\log_{5/4} L$ rounds, the maximal height of any tree is at most $3$.

    Beyond this point, if any tree has height $1$ after Step~(\ref{step_1}),
    then it must have height $1$ at the end of the previous round since there is no incident edge on leaves after the \textsc{alter} in the previous round,
    thus the loop must have been ended by the break condition.
    Therefore, the maximal-height tree (with height $3$ or $2$) cannot increase its height beyond this point.
    %(Readers can verify the height-$3$ or height-$2$ case if interested.)
    Suppose there is a tree with height $3$, then if the maximal level of vertices in this component does not change during the round, this tree cannot increase its height in the \textsc{maxlink} of Step~(\ref{step_6}) nor that of Step~(\ref{step_1}) in the next round,
    which means it has height at most $2$ as we do a \textsc{shortcut} in Step~(\ref{step_6}).
    So after $L$ rounds, all trees have heights at most $2$ after Step~(\ref{step_1}).
    After that, similarly, if the maximal level does not increase, all trees must be flat.
    Therefore, after additional $L$ rounds, all trees are flat after Step~(\ref{step_1}).
    By the same argument, the loop must have ended in the previous round.
    The lemmas follows immediately from $L = O(R)$ and Lemma~\ref{lem_dr2}.
\end{proof}

\subsection{Proof of Theorem~\ref{main3}} \label{sec_pf_main3}

With all pieces from \S{\ref{sec_correctness}}-\S{\ref{sec_diameter_reduction}} and Theorem~\ref{main1}, we are ready to prove Theorem~\ref{main3}.

\begin{proof}[Proof of Theorem~\ref{main3}]
    By Lemma~\ref{lem_flatten} and the break condition, the repeat loop in Faster Connected Components algorithm runs for $O(\log d + \log\log_{m/n} n) \le O(\log n)$ rounds.
    So by Lemma~\ref{lem_round_constant_time} and a union bound, this loop runs in $O(\log d + \log\log_{m/n} n)$ time with good probability.
    Moreover, since the diameter is $O(1)$ after the loop (cf. Lemma~\ref{lem_flatten}), the following Connected Component algorithm runs in $O(\log\log_{m/n} n)$ time with good probability (cf. Theorem~\ref{main1}).
    By Lemma~\ref{lem_stage_1}, the method \textsc{compact} also runs in $O(\log\log_{m / n} n)$ time with good probability.
    Therefore, by a union bound, the total running time is $O(\log d + \log\log_{m/n} n)$ with good probability.
    When it ends, by Lemma~\ref{lem_alg_correctness}, the algorithm correctly computes the connected components of the input graph.
    Since there are at most $O(\log n)$ rounds in the loop,
    by Lemma~\ref{lem_stage_1}, Lemma~\ref{lem_processor_number}, and Theorem~\ref{main1}, the total number of processors is $O(m)$ with good probability by a union bound.
    Theorem~\ref{main3} follows from a union bound at last.
\end{proof}

\addcontentsline{toc}{section}{Bibliography}
\bibliographystyle{alpha}
\bibliography{LTZ20}

\end{document}